\documentclass[prx,amsmath,amssymb, notitlepage, onecolumn,
nofootinbib,
superscriptaddress,
longbibliography
]{revtex4-1}

\usepackage{amsmath}
\usepackage{tabularx,graphicx}
\usepackage{epstopdf}
\usepackage{makecell}
\usepackage{graphicx}
\usepackage{latexsym}
\usepackage{amssymb}
\usepackage{amsthm}
\usepackage{color, colortbl}
\usepackage{psfrag}
\usepackage{bbm}
\usepackage{bm}
\usepackage{titlesec}
\usepackage{dsfont}
\usepackage{feynmp}
\usepackage{slashed}
\usepackage[utf8]{inputenc}
\usepackage[symbols,nogroupskip,sort=none]{glossaries-extra}
\usepackage{nomencl}

\usepackage{multirow}
\usepackage{url}
\usepackage{xr}
\usepackage{xcite}
\usepackage{microtype}

\makeatletter
\newcommand*{\addFileDependency}[1]{
  \typeout{(#1)}
  \@addtofilelist{#1}
  \IfFileExists{#1}{}{\typeout{No file #1.}}
}
\makeatother

\newcommand*{\myexternaldocument}[1]{%
    \externaldocument{#1}%
    \addFileDependency{#1.tex}%
    \addFileDependency{#1.aux}%
}

\myexternaldocument{main_05_12}

\newcommand{\beq}{\begin{eqnarray}}
	\newcommand{\eeq}{\end{eqnarray}}

\newcommand{\la}{\langle}
\newcommand{\ra}{\rangle}

\newcommand{\tr}{{\rm tr}}

\newcommand{\im}{{\rm Im}}
\newcommand{\bsp}{\begin{aligned}}
	\newcommand{\esp}{\end{aligned}}

\newcommand{\ie}{{i.e., }}
\newcommand{\eg}{{e.g., }}

\newcommand{\rH}{\mathrm{H}}
\newcommand{\R}{\mathbb{R}}
\newcommand{\Aut}{\mathrm{Aut}}

\usepackage{xcolor}
\definecolor{darkblue}{rgb}{0.,0.,0.4}
\definecolor{darkred}{rgb}{0.5,0.,0.}
\definecolor{BlueViolet}{RGB}{138,43,226}
\definecolor{SkyBlue}{RGB}{30,144,255}
\definecolor{DarkGreen}{RGB}{0,100,0}

\usepackage[colorlinks=true,linkcolor=darkblue,citecolor=blue,urlcolor=darkred]{hyperref}
\usepackage[normalem]{ulem}

\usepackage{mathrsfs}

\usepackage{comment}

\newcommand{\z}{\mathbb{Z}}

\newcommand{\A}{\mathscr{A}}
\newcommand{\B}{\mathcal{B}}
\newcommand{\G}{\mathcal{G}}
\newcommand{\cH}{\mathcal{H}}
\newcommand{\cU}{\mathcal{U}}
\newcommand{\Ad}{\mathrm{Ad}}
\newcommand{\cC}{\mathcal{C}}
\newcommand{\bbC}{\mathbb{C}}

\newcommand{\QCA}{\mathrm{QCA}}
\newcommand{\ind}{\mathrm{ind}}

\newcommand{\tA}{\widetilde{\Aut}}
\newcommand{\tG}{\widetilde{\G}}
\newcommand{\id}{\mathrm{id}}

\newcommand{\Vc}{\textbf{Vec}_{\bbC}}
\newcommand{\ad}{\mathrm{ad}}
\newcommand{\talpha}{\tilde{\alpha}}

\newcommand{\cK}{\mathcal{K}}
\newcommand{\cM}{\mathscr{M}}

\newcommand{\diam}{\mathrm{diam}}
\def\U{\mathrm{U}(1)}

\makeatletter
\newcommand{\colim@}[2]{%
  \vtop{\m@th\ialign{##\cr
    \hfil$#1\operator@font colim$\hfil\cr
    \noalign{\nointerlineskip\kern1.5\ex@}#2\cr
    \noalign{\nointerlineskip\kern-\ex@}\cr}}%
}
\newcommand{\colim}{%
  \mathop{\mathpalette\colim@{\rightarrowfill@\textstyle}}\nmlimits@
}
\makeatother

\newtheorem{corollary}{Corollary}
\newtheorem{theorem}{Theorem}
\newtheorem{lemma}{Lemma}
\newtheorem{definition}{Definition}
\newtheorem{example}{Example}
\newtheorem{proposition}{Proposition}
\newtheorem{remark}{Remark}
\makenomenclature

\numberwithin{equation}{section}
\numberwithin{corollary}{section}
\numberwithin{theorem}{section}
\numberwithin{lemma}{section}
\numberwithin{definition}{section}
\numberwithin{example}{section}
\numberwithin{proposition}{section}
\numberwithin{remark}{section}

\usepackage{pst-all}

\usepackage{leftidx}

\usepackage[off]{auto-pst-pdf} 
\usepackage{booktabs}
\usepackage{yfonts}
\makeatletter

\usepackage[caption=false]{subfig}
\usepackage{enumitem}

\begin{document}

\title{Twisted locality-preserving automorphisms, anomaly index,\\
and generalized Lieb-Schultz-Mattis theorems with anti-unitary symmetries}

\author{Ruizhi Liu}
\affiliation{Perimeter Institute for Theoretical Physics, Waterloo, Ontario, Canada N2L 2Y5}
\affiliation{Department of Mathematics and Statistics, Dalhousie University, Halifax, Nova Scotia, Canada, B3H 4R2}

\author{Jinmin Yi}
\affiliation{Perimeter Institute for Theoretical Physics, Waterloo, Ontario, Canada N2L 2Y5}
\affiliation{Department of Physics and Astronomy, University of Waterloo, Waterloo, Ontario, Canada N2L 3G1}

\author{Liujun Zou}
\affiliation{Perimeter Institute for Theoretical Physics, Waterloo, Ontario, Canada N2L 2Y5}
\affiliation{Department of Physics, National University of Singapore, 117551, Singapore}

\begin{abstract}
    Symmetries and their anomalies are powerful tools to understand quantum matter. In this work, for quantum spin chains, we define twisted locality-preserving automorphisms and their Gross-Nesme-Vogts-Werner indices, which provide a unified framework to describe both unitary and anti-unitary symmetries, on-site and non-on-site symmetries, and internal and translation symmetries. For a symmetry $G$ with actions given by twisted locality-preserving automorphisms, we give a microscopic definition of its anomaly index, which is an element in $H^3_\varphi(G; U(1))$, where the subscript $\varphi$ means that anti-unitary elements of $G$ act on $U(1)$ by complex conjugation. We show that an anomalous symmetry leads to multiple Lieb-Schultz-Matttis-type theorems. In particular, any state with an anomalous symmetry must either have long-range correlation or violate the entanglement area law. Based on this theorem, we further deduce that any state with an anomalous symmetry must have long-range entanglement, and any Hamiltonian that has an anomalous symmetry cannot have a unique gapped symmetric ground state, as long as the interactions in the Hamiltonian decay fast enough as the range of the interaction increases. For Hamiltonians with only two-spin interactions, the last theorem holds if the interactions decay faster than $1/r^2$, where $r$ is the distance between the two interacting spins. We demonstrate these general theorems in various concrete examples.
\end{abstract}

\maketitle

\tableofcontents\

\section{Introduction}

One of the central tasks of quantum many-body physics is to characterize and classify quantum phases of matter. Symmetry is a key ingredient in this endeavor. The most familiar example where symmetry plays a crucial role is spontaneous symmetry breaking (SSB), where the ground states fail to respect the symmetry of the underlying Hamiltonian. However, the use of symmetry in modern quantum many-body physics extends far beyond SSB. In particular, it is by now well established that the anomalies associated with symmetries provide powerful, non-perturbative constraints on the correlation, entanglement and energy spectrum of quantum many-body systems \cite{Hooft1980, Chen2013, Wen2013}.

A paradigmatic example of such a constraint is given by the Lieb–Schultz–Mattis (LSM) theorems \cite{Lieb1961, Oshikawa1999, Hastings2003}. These fundamental theorems have become a cornerstone in our understanding of strongly correlated systems, with wide-ranging implications for quantum criticality, quantum spin liquids, and symmetry-protected topological (SPT) phases. Recently, the LSM constraints have been interpreted from various perspectives and generalized to different contexts~\cite{Cheng2015, Po2017, Jian2017,  Cho2017,Watanabe2018LSM,Metlitski2018,Cheng2018a, Kobayashi2018, Ogata2019LSM,Else2020, Jiang2019,Yao2021twisted, Ogata_2021, Aksoy2021, Ye2021a,Ma2022a, Cheng2022, Kawabata2023, Aksoy2023, Seifnashri2023, Zhou2023, kapustin2024anomalous,Garre_Rubio2024anomalous, Liu2024LRLSM, Pace2024}. Furthermore, these constraints are identified as a key ingredient to study the classification of quantum phases of matter in a lattice system~\cite{Zou2021, Ye2021a, Ye2023, Liu2024}.

In addition to symmetry and its associated anomaly, another crucial ingredient of quantum many-body physics is locality, which sharply distinguishes a quantum many-body system from ordinary quantum mechanical systems. The assumption of locality-preserving dynamics is essential for the validity of LSM-type theorems, and the action of symmetries must preserve a notion of spatial locality. Mathematically, this requirement is captured by the framework of quantum cellular automata (QCA) or locality-preserving automorphisms (LPA) \cite{Gross_2012, Ranard_2022, kapustin2024anomalous,Tu2025}. These structures formalize the idea that time evolutions and symmetry transformations act on quantum operators in a way compatible with locality.

However, the standard definitions of QCA and LPA are limited to unitary symmetries. Importantly, they exclude anti-unitary symmetries, such as time reversal, which play a central role in many physical systems. To extend the LSM-type constraints to include such cases, it is worthwhile generalizing the framework of QCA and LPA to accommodate anti-unitary symmetry actions. In this work, we introduce twisted versions of QCA and LPA in quantum spin chains, which incorporate anti-unitary symmetries into the locality-preserving paradigm. Building on this formalism, we define a corresponding anomaly index from the twisted symmetry action and establish new LSM-type constraints. Our results provide a unified perspective on how both unitary and anti-unitary symmetries enforce nontrivial constraints on quantum matter, thereby broadening the reach of the LSM constraints.

The rest of the paper is organized as follows. In Sec. \ref{sec: background and terminology}, we briefly review the operator algebra formalism extensively used later in the paper. In Sec. \ref{sec: twisted LPA}, we define twisted locality-preserving automorphisms and their Gross-Nesme-Vogts-Werner indices, which provide a unified framework to discuss both unitary and anti-unitary symmetries. Next, in Sec. \ref{sec:construction_anomaly_index}, for any symmetry described by twisted locality-preserving automorphisms, we define its anomaly index microscopically. With the notion of the anomaly index, we explore the consequences of an anomalous symmetry in Sec. \ref{sec: LSM infinite chains}. We present multiple LSM-like theorems regarding the interplay between symmetries, correlations, entanglement and energy spectra of quantum many-body systems, summarized in Lemma \ref{lemma: root lemma}, Corollary \ref{corollary: cluster + area law}, Corollary \ref{corollary: Kapustin-Sopenko}, Theorem \ref{thm: spectrum} and Theorem \ref{thm: LRE}. Up to this section, all discussions are about quantum spin chains with an infinite size. To extract consequences on finite-size quantum spin chains, which are physically relevant, in Sec. \ref{sec: finite systems} we carefully define the thermodynamic limits of a sequence of finite systems, which connect finite and infinite systems. Then we derive the finite-size versions of Theorem \ref{thm: spectrum} and Theorem \ref{thm: LRE}, which are the more standard versions of the LSM theorems. After establishing these general results, we demonstrate their power in various examples in Sec. \ref{sec: applications}. We end this paper with discussions and outlooks in Sec. \ref{sec: discussion}. Various appendices contain additional technical details.

\section{Review of the operator algebra formalism} \label{sec: background and terminology}

Our discussion relies on the formalism of operator algebra, which can conveniently deal with both finite-size and infinite-size systems. In this section, we briefly review the background and terminology of this formalism that are necessary to formulate and prove our results. Interested readers can find more detailed reviews in, for example, Appendix A of Ref. \cite{Liu2024LRLSM} and the references therein.

\subsection{Algebras of local and quasi-local operators}\label{subsec:quasi_local}

We first introduce some general background of the algebras of local operators and quasi-local operators. Readers are referred to Refs. \cite{bratteli2013operator1,bratteli2013operator2,Naaijkens_2017,Landsman:2017hpa} for a more thorough treatment.

Throughout this subsection, we work on lattices of general spatial dimension $d$, \ie for an infinite-size system our lattice is $\Lambda\simeq \z^{d}$ unless otherwise specified. We assume that our on-site Hilbert space $\cH_{k}$ is finite-dimensional, where $k$ labels a site in $\Lambda$ (we do {\it not} assume that $\dim\cH_{k}$ is the same for different $k$'s). We emphasize that the notion of the ``total Hilbert space" of an infinite-size system does not make sense, because there is no well defined inner product between ``state vectors" in this space (see, for example, Appendix A of Ref. \cite{Liu2024LRLSM} for more detail).

Nevertheless, the Heisenberg picture of quantum mechanics, which focuses on operators rather than states, is still applicable even to systems with infinitely many degrees of freedom, such as spin systems on infinite lattices. In this formalism, we start with the notion of {\it{local operators}}. Given a {\it finite subset} $\Gamma\subset\Lambda$, one can talk about the Hilbert space $\cH_{\Gamma}:=\bigotimes_{k\in \Gamma}\cH_{k}$ on $\Gamma$. The operator supported on $\Gamma$ is defined to be all operators on this finite dimensional Hilbert space $\cH_{\Gamma}$, including $c$-numbers. Note that any (finite) addition and multiplication of operators on $\cH_{\Gamma}$ give another operator on $\cH_{\Gamma}$, and hence these operators form an algebra, denoted by $\A_{\Gamma}$. We call $\A_{\Gamma}$ the algebra of local operators support on $\Gamma$. It is obvious that if $A\in\A_{\Gamma}$, then its Hermitian conjugate $A^{\dagger}\in\A_{\Gamma}$ as well. 

Now we introduce a norm on the algebra $\A_\Gamma$. For a local operator $A\in\A_{\Gamma}$, we define its operator norm by 
\beq\label{eq:operator_norm}
||A||:=\sup_{\substack{|\psi\ra\in\cH_{\Gamma},\\ \la\psi|\psi\ra=1}}|A|\psi\ra|,
\eeq
where $\langle\psi|\psi\rangle$ is the ordinary norm squared of a vector $|\psi\rangle\in\cH_\Gamma$. According to this definition, the norm of the operator $A$ is the square root of the largest eigenvalue of $A^\dag A$ in the case with only finitely many degrees of freedom.

Given any local operator $A\in \A_{\Gamma}$, if $\Gamma'$ is another finite subset containing $\Gamma$ (\ie $\Gamma\subset \Gamma'$), then there is a natural way to extend $A$ to a local operator supported on $\Gamma'$,
\beq
\tilde{A}=A\bigotimes_{k\in\Gamma'\setminus \Gamma}I_{k},
\eeq
where $\tilde{A}$ is the extension of $A$ on $\Gamma'$ and $I_{k}$ is the identity operator on $\cH_{k}$. In this case, we say that $A$ acts as an identity outside $\Gamma$.
It is often convenient to identify these two operators. One can easily check that $||\tilde{A}||=||A||$, so this identification is unambiguous for norms.

We then make the following definition:

\begin{definition}\label{def:local_operators}
The algebra of local operators
    \beq\label{eq:local_operator_algebra}
\A^{l}:=\bigcup_{\Gamma\subset\Lambda,|\Gamma|<\infty}\A_{\Gamma}
\eeq
with the above identification $\tilde{A}\sim A$. We also define $\A_{\emptyset}=0$.
\end{definition}
Here $|\Gamma|$ means the cardinality of $\Gamma$, and we write $|\Gamma|<\infty$ if $\Gamma$ is a finite set. More explicitly, $A\in \A^{l}$ if there is a finite subset $\Gamma$ such that $A\in \A_{\Gamma}$ (this $\Gamma$ is not unique due to the freedom to extend $A$).

One crucial property of $\A^{l}$ is locality. Given $A_{i}\in \A_{\Gamma_{i}},i=1,2$, if $\Gamma_{1}\cap \Gamma_{2}=\emptyset$, then
\beq
[A_{1},A_{2}]=0.
\eeq
We remark that $A_{1,2}$ in above equation really mean their extensions on some $\Gamma$ that contains $\Gamma_{1}\cup\Gamma_{2}$.

Although the above definition of the local operator algebra $\A^l$ is clear and physically motivated, it is often insufficient to only work with $\A^{l}$. To see it, consider a time evolution generated by a local Hamiltonian {\footnote{For an infinite-size system, this way of writing the Hamiltonian is not rigorous, but the physical picture here is valid. The more rigorous definition of the Hamiltonian in this case will be introduced in Sec. \ref{subsec: Hamiltonians and ground states}.}}
\beq\label{eq:charge_densities}
H=\sum_{|\Gamma|<\infty}h_{\Gamma},
\eeq
where $h_{\Gamma}$ is some local term, \ie $h_{\Gamma}=0$ if $\diam(\Gamma)>R$ for some fixed $R>0$, with $\diam(\Gamma):=\max_{x,y\in\Gamma}d(x,y)$ where $d$ is the distance. Given such an $H$, the time evolution is denoted by $\alpha^{t}$ for $t\in\R$, which formally transforms a local operator $A\in \A_{\Gamma}$ into
\beq\label{eq:symLPA}
\alpha^{t}(A)=e^{it H}A e^{-it H}
\eeq
As noted above, for a given $A\in\A_\Gamma$, if $B\in \A_{\Gamma'}$ is another local operator such that $\Gamma\cap \Gamma'=\emptyset$, then $[A,B]=0$ by locality. However, $[\alpha^{t}(A),B]\not=0$ in general and this commutator is bounded by the famous Lieb-Robinson bound \cite{Lieb:1972wy,Matsuta2017LR,Else2018LR,_Anthony_Chen_2023}
\beq\label{eq:LR_bound}
||[\alpha^{t}(A),B]||\leqslant |\Gamma| C e^{-a(L-vt)}
\eeq
where $C,a,v$ are non-universal positive constants and $L:=d \, (\Gamma,\Gamma')$  is the distance between $\Gamma$ and $\Gamma'$. The constant $v$ is usually called the Lieb-Robinson velocity. The above bound means that although $\alpha_t(A)$ is generically not a truly local operator, it is still approximately a local operator.

Therefore, to ensure that time evolutions can act sensibly on our operator algebra, one has to consider the {\it{quasi-local}} operator algebra $\A^{ql}$ (see below for definition), rather than just $\A^{l}$. Intuitively, $\A^{ql}$ is obtained by taking sequential limits in $\A^{l}$. Concretely, an operator $A\in \A^{ql}$ if and only if there is a Cauchy sequence $A_{j}\in\A^{l}$ such that 
\beq
A=\lim_{j\to \infty}A_{j}.
\eeq
By a Cauchy sequence, we mean that $\forall\,\epsilon>0$, there exists $N\in\z^{>0}$ such that $||A_{j}-A_{j'}||<\epsilon$ for all $j,j'>N$, where $||\cdot||$ is the operator norm defined in Eq. \eqref{eq:operator_norm}. More precisely, one says that $\A^{ql}$ is the completion of $\A^{l}$ with respect to $||\cdot||$.

In essence, if $A\in\A^{ql}$, $A$ may not act as identity outside a finite region. However, it can be approximated by a certain local operator given any desired accuracy. Moreover, $\A^{ql}$ has a natural norm inheriting from the operator norm in Eq. \eqref{eq:operator_norm} on $\A^{l}$. Explicitly, let $A\in\A^{ql}$ and $A_{j}\in\A^{l}$ be a sequence convergent to $A$, we define
\beq\label{eq:limit_norm}
||A||:=\lim_{j\to\infty}||A_{j}||.
\eeq
It is easy to check that $||A||$ does not depend on the choice of the sequence $\{A_{j}\}_{j=1,2,...}$. By this definition, the unit in $\A^{ql}$, \ie the identity operator $I$, has norm 1. 

To summarize,

\begin{definition}\label{def:quasi-local}
    The quasi-local operator algebra $\A^{ql}$ is defined to be the completion of $\A^{l}$ with respect to the operator norm in Eq. \eqref{eq:operator_norm}. We also denote the group of quasi-local unitary operators as $\cU^{ql}$.
\end{definition}

Given a subset $\Gamma$ (which may be finite or infinite), we write $\A^{ql}_\Gamma$ for the algebra of quasi-local operators supported on $\Gamma$, \ie it acts as identity outside of $\Gamma$.

It is worth noting that our $\A^{ql}$ is a special example of the so-called $C^{*}$-algebras in the mathematical literature \cite{davidson1996c,arveson1998invitation,murphy2014c}, which are defined below. 
\begin{definition}\label{def:C*_alg}
    A $C^{*}$-algebra $\cC$ is an algebra equipped with an anti-linear involution $*$ (which models Hermitian conjugation in quantum mechanics) and a norm $||\cdot||$ (which may or may not be the operator norm defined in Eq. \eqref{eq:operator_norm}), such that
\begin{enumerate}
    \item $(A^{*})^*=A$ for all $A\in \cC$.
    \item $(AB)^{*}=B^{*}A^{*}$ for all $A,B\in \cC$.
    \item $(\lambda A+B)^{*}=\bar{\lambda}A^{*}+B^{*}$ for all $A,B\in \cC$, $\lambda\in \bbC$ and $\bar{\lambda}$ is the complex conjugate of $\lambda$.
    \item (Banach property) $||AB||\leqslant ||A||\cdot||B||$ and $||A^{*}||=||A||$ for all $A,B\in\cC$.
    \item ($C^{*}$ property) $||A^*A||=||A||^{2}$
\end{enumerate}
Besides, $\cC$ should be complete with respect to the norm $||\cdot||$, \ie any Cauchy sequence of elements in $\cC$ converges into an element in $\cC$.
\end{definition}

In this paper, we will restrict ourselves to two important examples of $C^*$-algebra, \ie the algebra of quasi-local operators $\A^{ql}$ and the algebra of bounded operators $\B(\cH)$ in a Hilbert space $\cH$.

\subsection{Quantum cellular automata and locality-preserving automorphisms}\label{subsec:QCA_LPA}

After introducing the basic notions of operator algebras, our next goal is to define a proper notion of symmetry action on local operators. This symmetry action is required to preserve certain notion of locality. In this subsection, we focus on unitary symmetries. The discussion will be generalized to include both unitary and anti-unitary symmetries in Sec. \ref{sec: twisted LPA}.

Given a symmetry group $G$, it is natural to define the symmetry action as a homomorphism $G\to \Aut(\A^{l})$, where $\Aut(\A^{l})$ is the automorphism group of the local operator algebra $\A^{l}$, which is defined below.
\begin{definition} \label{def: automorphisms}
    We say that $\alpha:\A^{l}\to\A^{l}$ is an automorphism of $\A^{l}$, if
\begin{enumerate}
    \item $\alpha(A+B)=\alpha(A)+\alpha(B)$
    \item $\alpha(AB)=\alpha(A)\alpha(B)$
    \item $\alpha(A^{\dagger})=\alpha(A)^{\dagger}$
    \item $\alpha(\lambda A)=\lambda\alpha(A),\forall\,\lambda\in\bbC$
    \item $\alpha$ is invertible.
\end{enumerate}
All automorphisms of $\A^{l}$ form a group under \textbf{finite} compositions, denoted by $\Aut(\A^{l})$
\end{definition}

There is a special subgroup of $\Aut(\A^{l})$ in the literature called quantum cellular automata \cite{arrighi2019overview,Farrelly_2020}, which will be useful to describe symmetry actions.

\begin{definition}\label{def:QCA}
    A quantum cellular automaton (QCA) is an automorphism $\alpha$ of $\A^{l}$ such that for a local operator $A\in\A_{X}$ (where $X$ is a finite subset), $\alpha(A)\in\A_{B(X,r_{\alpha})}$, where $r_{\alpha}>0$ does not depend on $A$.
\end{definition}
Here $B(\Gamma,r_{\alpha}):=\{x\in\Lambda|d(x,\Gamma)\leqslant r_{\alpha}\}$, where $d$ is the distance on lattice and $0\leqslant r_{\alpha}<\infty$ does not depend on $A$. According to this definition, QCA preserves the locality of an operator. Moreover, QCA form a group under finite compositions, which is denoted by $\G^{\QCA}$.
\nomenclature{$\G^{QCA}$}{The group of quantum cellular automata, def. \ref{def:QCA}} 

It is useful to look at a special example of QCA, which will also be discussed later.

\begin{example}\label{example:circuit}

    One particular example of QCA is a (finite-depth unitary) circuit. For simplicity, we describe it in one spatial dimension (1d). Let $\{P_{k}\}_{k\in\z}$ be a set of disjoint intervals with $|P_{k}|<l$ for a constant length $l$ ($P_{k}$ can be empty for some $k$'s). Then we define a block-partitioned unitary (BPU) as 
    \beq\label{eq:BPU}
    \alpha=\prod_{k=-\infty}^{\infty}\Ad_{U_{k}}
    \eeq
    where $U_{k}$ is a unitary operator supported on $P_{k}$ and $\Ad_{U_{k}}(A)=U_{k}AU_{k}^{-1}$ for local operator $A$. One can check that a BPU is a QCA using Definition \ref{def:QCA}.

    A circuit is a finite composition of BPU's (these BPU's may be defined for different partitions). The group of all circuits are denoted by $\G^{cir}$. Later we will see that not every QCA is a circuit.
    
\end{example}
\nomenclature{$\G^{cir}$}{The group of circuits under finite composition, example \ref{example:circuit}}

As explained in the last subsection, one needs to consider not only QCA's. Instead, the automorphism group $\Aut(\A^{ql})$ of the quasi-local operator algebra $\A^{ql}$ is also of fundamental importance.

\begin{definition}\label{def:LPA}
    An automorphism $\alpha$ of $\A^{ql}$ is called a locality-preserving automorphism if for each local operator $A\in\A_{X}$ and any $r>0$, there exists a local operator $B$ such that
    \beq
    ||\alpha(A)-B||<f_{\alpha}(r)||A||
    \eeq
    where $f_{\alpha}(r)$ is a positive decreasing function independent of the choice of $A$ and $\lim_{r\to \infty}f_{\alpha}(r)=0$.

    The group of LPA's under finite composition is denoted by $\G^{lp}$.
\end{definition}
\nomenclature{$\G^{lp}$}{Group of locality-preserving automorphism}
More explicitly, if $\alpha\in \G^{lp}$ and $A\in\A_{\Gamma}$ is a local operator with $|\Gamma|<\infty$, although $\alpha(A)$ is {\it not} a local operator any more in general, it can be approximated by another local operator $B$ defined on a larger support $B(\Gamma,r)$ with an error controlled by $f_{\alpha}(r)$. Clearly, all QCA are LPA.
Theorem 3.2 and Lemma 3.3 of Ref.~\cite{Ranard_2022} give an easy criterion for an automorphism to be an LPA in 1d. In particular, a finite-time evolution generated by a local Hamiltonian is an LPA.

We will consider a very broad class of symmetry actions, which are descried by LPA. Concretely, the symmetry action is in general a homomorphism from the symmetry group $G$ to $\G^{lp}$. In order to study these symmetry actions on lattice systems in more detail, one first needs the structure of $\G^{lp}$. In fact, the structures of $\G^{\QCA}$ and $\G^{lp}$ are rather clear in 1d,\footnote{Some classifications of higher dimensional QCA's are also proposed recently, see \eg Ref. \cite{Freedman_2020}.} thanks to the Gross-Nesme-Vogts-Werner (GNVW) index \cite{Gross_2012,Ranard_2022}. Here we only give a brief introduction to the GNVW index below without defining the index explicitly. The main idea of the construction and the precise definition of the GNVW index are reviewed in Appendix \ref{sec:QCA_review}. Interested readers are referred to Refs. \cite{Gross_2012,arrighi2019overview,Farrelly_2020,Ranard_2022} for more details.

We start with the GNVW index of a QCA. We assume that all on-site local Hilbert spaces $V$ have the same dimension $D$. This does not lose any generality if the dimensions of the local Hilbert spaces are uniformly upper bounded, since in this case equal dimension for all local Hilbert spaces can always be achieved by tensoring the original degrees of freedom with some other degrees of freedom at each site, such that our QCA acts trivially on the additional degrees of freedom.

Roughly speaking, the GNVW index is a group homomorphism, denoted by
\beq\label{eq:GNVW_index}
\ind:\G^{\QCA}\to \z[\{\log(p_{j}) \}_{j\in J}],
\eeq
\nomenclature{$\ind$}{The GNVW index map, \eqref{eq:GNVW_index}}
where $\{p_{j}\}_{j\in J}$ is the set of all prime divisors of $D$. This index map satisfies the following properties:
\begin{enumerate}
    \item $\ind(\alpha\beta)=\ind(\alpha)+\ind(\beta),\forall\,\alpha,\beta\in\G^{\QCA}$,
    \item $\ker(\ind)=\G^{cir}$,
    \item $\ind$ is a surjection.
\end{enumerate}
It can be verified that a circuit is a QCA with a vanishing GNVW index. On the other hand, one can verify that if $\tau$ is a shift (\ie translation) on the lattice by $+1$ unit cell, then
\beq\label{eq:Generalized_translations}
\ind(\tau)=\log D\in \z[\{\log p_{j}\}_{j\in J}].
\eeq
\nomenclature{$\G^{T}$}{Group of generalized translations, \eqref{eq:Generalized_translations}}

In order to understand that $\ind$ is a surjection, note that $\z[\{\log p_{j}\}_{j\in J}]$ is generated by $\log p_i$, so it suffices to find an element in $\G^{\QCA}$ for each $\log p_i$, such that the index of this element under $\ind$ is exactly $\log p_i$. To this end, one defines a generalized translation (or partial translation), which only shifts part of the degrees of freedom at each site and which has an index $\log p_i$. To be more precise, one can fix a $p_{i}$ (where $p_{i}$ is a prime divisor of $D$) dimensional subspace $V_{i}$ of $V$. One can then factorize $V$ into $V\simeq V'\otimes V_{i}$ at each site. A generalized translation $\tau_{i}$ only shifts the $V_{i}$-part while fixing $V'$. Using the definition of the GNVW index, the GNVW index of $\tau_i$ is verified to be
\beq
\ind(\tau_{i})=\log p_{i}.
\eeq

Generalized translations form an Abelian group under finite compositions, and we denote it by $\G^{T}$. Thus, the GNVW index actually shows that $\G^{\QCA}$ is a semi-direct product
\beq\label{eq:QCA_semi_direct}
\G^{\QCA}=\G^{cir}\rtimes \G^{T}.
\eeq

After the above review on the structures of $\G^{\QCA}$, now we turn to the structure theory of LPA. Recall that $\A^{ql}$ is obtained by taking limits of $\A^{l}$, or more formally, any quasi-local operator can be approximated by local operators. Similarly, any LPA can also be approximated by QCA. The basic strategy to study LPA is to approximate it by a sequence of QCA. For example, if $\alpha\in\G^{lp}$, there exists a sequence $\{\beta_{j}\}_{j=1,2...}$ of QCA such that
\beq
\lim_{j\to\infty}\beta_{j}=\alpha.
\eeq
Then one defines the GNVW index of $\alpha$ as
\beq
\ind(\alpha)=\lim_{j\to\infty}\ind(\beta_{j}).
\eeq
It can be checked that this index is well-defined (finite and independent of the choice of the sequence) \cite{Ranard_2022}. Given the existence of this index map, many results of QCA can be carried over to LPA. For example, 
\beq
\G^{lp}=\G^{loc}\rtimes\G^{T}
\eeq
where $\G^{loc}$ is the subgroup of LPA with GNVW index 0, which are time evolution generated by local Hamiltonians that can have arbitrary time dependence (\eg such Hamiltonians may not be a continuous function of time). Index-0 LPA are called ``local Hamiltonian evolutions" in Refs.~\cite{Ranard_2022,kapustin2024anomalous}. This terminology might be confusing, since usually we assume that the time evolutions are differentiable in time (\ie the Hamiltonian is continuous in time), but $\G^{loc}$ includes time evolutions that are not differentiable in time.

\subsection{States, GNS construction and factors}

After introducing the notions of local and quasi-local operator algebras, QCA and LPA, now we introduce the notions of states, Gelfand-Naimark-Segal construction and von Neumann factors in this subsection. These concepts are also necessary for formulating and proving our results.

A state $\psi$ of the $C^*$-algebra $\A^{ql}$ is defined to be a nonzero linear functional $\psi:\A^{ql}\to \bbC$ which is
\begin{enumerate}
    \item Positive, \ie $\psi(A^{\dagger}A)\geqslant 0,\,\forall\,A\in\A^{ql}$.
    \item Normalized, \ie $\psi(I)=1$, where $I$ is the identity operator in $\A^{ql}$.
\end{enumerate}
Given two quantum states $\psi_{0},\psi_{1}$ and a real number $t\in [0,1]$, one can construct another state
\beq\label{eq:convex_combination}
\psi_{t}=t\psi_{1}+(1-t)\psi_{0}.
\eeq
It is easy to check that $\psi_{t}$ is positive and normalized, thus indeed being a valid state. A state $\psi$ is said to be mixed if there exists $\psi_{0},\psi_{1}$ and $0<t<1$ such that $\psi=\psi_{t}$. A state is pure if it is not mixed.

For any state $\psi$, there is a representation of $\A^{ql}$ associated with it, which is given by the so-called Gelfand-Naimark-Segal (GNS) construction. The data of this representation is encapsulated in the GNS triple $(\cH_{\psi},\pi_{\psi},|\psi\ra)$, where $\cH_{\psi}$ is a Hilbert space, $\pi_{\psi}:\A^{ql}\to\B(\cH_{\psi})$ is a $*$-homomorphism and $|\psi\ra$ is a vector state which represents the abstract state $\psi$ on $\cH_{\psi}$, in the sense that
\beq
\psi(A)=\la\psi|\pi_{\psi}(A)|\psi\ra, \forall A\in\A^{ql}.
\eeq
It can be shown that the GNS triple is unique up to unitary equivalence (see, for example, Ref.~\cite{Naaijkens_2017}). If two representations $\pi_{1}$ and $\pi_{2}$ are unitarily equivalent, then we write $\pi_{1}\simeq\pi_{2}$. A GNS representation of $\psi$ is irreducible if and only if the state $\psi$ is pure \cite{Naaijkens_2017}.

Next, we give a brief introduction to von Neumann factors and factor states, which are special states that are of physical interest.
Familiarity with von Neumann algebras would be helpful, but it is not necessary for the following discussions. See, \eg the Appendix A of Ref.~\cite{Liu2024LRLSM} for an introduction to von Neumann algebra for working physicists.

We begin with the definition of factor representations \cite{kadison1997fundamentals1,Landsman:2017hpa}.
\begin{definition} \label{definition: factor}
    Let $\pi:\A^{ql}\to\B(\cH)$ be a $*$-representation of $\A^{ql}$. It is called irreducible if the commutant
    \beq
    \pi(\A^{ql})':= \{x\in \B(\cH)| [x,\pi(\A^{ql})]=0\}
    \eeq
    equals $\mathbb{C}\cdot\id_{\cH}$. It is called a factor (also known as primary) representation if
    \beq
    \pi(\A^{ql})'\cap (\pi(\A^{ql})')' = \mathbb{C}\cdot \id_\cH
    \eeq
    A state $\psi$ of $\A^{ql}$ is called a factor state if its GNS representation is a factor representation.
\end{definition}

The physical importance of factor states is recognized by the following profound theorem.
\begin{theorem}[Theorem 2.6.10 of Ref. \cite{bratteli2013operator1}]\label{thm:clustering}
    A state $\psi$ on $\A^{ql}$ is clustering iff it is a factor state.
\end{theorem}
Here clustering means that the connected correlation function of any two local operators decays to zero (no matter how fast) when the distance between these two local operators is taken to infinity. The above theorem is elegant because it connects the clustering property, which is an important and intuitive concept related to the locality structure of the states, with the properties of the GNS representation of the state, which appears to be much more abstract.

Mathematically, there is a classification of factors into three types (see Refs.~\cite{bratteli2013operator1,Landsman:2017hpa} for details about this classification). The only relevant one for us will be type-I factors.
\begin{definition}
    Let $\pi$ be a factor representation of $\A^{ql}$. It is called type-I if the double commutant
    \beq
    \pi''(\A^{ql}):= (\pi(\A^{ql})')'\simeq\B(\cK)
    \eeq
    for some Hilbert space $\cK$. Similarly, a state $\psi$ is called type-I if its GNS representation is type-I.
\end{definition}
\begin{remark}
    Given a state $\psi$ (not necessarily a factor), the double commutant associated to $\psi$ will be denoted by $\cM_{\psi}$. It turns out this algebra $\cM_{\psi}$ is a so-called von Neumann algebra, and we will call it as the von Neumann algebra associated with $\psi$.
\end{remark}

The reason we need type-I states is that they are particularly well-behaved. More concretely,
\begin{lemma}[Lemmas A.5, A.7 of Ref.~\cite{Liu2024LRLSM}]\label{lemma:typeI}
    Let $\psi$ be a type-I factor and $\cM_{\psi}$ be the double commutant algebra. Then
    \begin{itemize}
        \item $\psi=\sum_{k=1}^{\infty}\lambda_k\psi_k$, where $0\leqslant\lambda_k\leqslant1$ with $\sum_k\lambda_k=1$ and $\psi_k$ are pure states such that $\psi_i\simeq\psi_j$ for any $i,j$. In particular, all pure states are type-I.
        \item All $*$-automorphisms of $\cM_{\psi}$ are inner, \ie they are given by $\Ad_U$ for some $U\in \cM_{\psi}$.
    \end{itemize}
\end{lemma}
These two properties will be extremely useful when proving a root lemma, Lemma \ref{lemma: root lemma}, based on which all our results are derived.

When is a factor state of type-I? Sometimes the definition can be awkward to use and an alternative criterion might be useful. Ref.~\cite{Liu2024LRLSM} provides a useful and convenient way (see Proposition A.10 therein):
\begin{lemma} [Proposition A.10 of Ref~\cite{Liu2024LRLSM}]\label{lemma: area law implies type I}
    Let $\psi$ be a factor state of a 1d spin chain satisfying the area law of entanglement entropy\footnote{Given a state $\psi$ in 1d and a finite interval $[a, b]$, the restriction $\psi|_{[a,b]}$ can be represented by a density matrix $\rho$. The state $\psi$ satisfies the area law if the entanglement entropy $S:=-\tr(\rho\log\rho)<D$ for some $D>0$ that is independent of $a$ and $b$.}, then $\psi$ must be of type-I.
\end{lemma}

We finish this subsection by briefly discussing the relations between two states. Given two states, we are often interested in how different they are, which can be characterized by the following concepts.
\begin{definition} \label{def: quasi-equivalence}
    Let $\psi,\rho$ be two states of $\A^{ql}$, they are called quasi-equivalent if there is a $*$-isomorphism $f:\cM_{\psi}\to\cM_{\rho}$ such that
    \beq
    f(\pi_{\psi}(a))=\pi_{\rho}(a),\forall a\in \A^{ql}.
    \eeq
    We write $\psi\sim\rho$ for quasi-equivalence.
    
    They are called disjoint if their GNS representation contains no equivalent subrepresentations.
\end{definition}
In general, two states can be neither quasi-equivalent nor disjoint, but the situation significantly simplifies for factor states. In fact,
\begin{lemma}[Corollary 8.22 of Ref.~\cite{Landsman:2017hpa}]
    Two factor states are either disjoint or quasi-equivalent.
\end{lemma}

\subsection{Split property}

To state and prove our results, the next ingredient is the split property of factor states. Recall that we always have the factorization of the algebra, $\A^{ql}\simeq\A^{ql}_{<0}\otimes\A^{ql}_{\geqslant0}$, where $\A^{ql}_{<0}$ (resp. $\A^{ql}_{\geqslant 0}$) is the algebra of quasi-local operators supported on the negative half chain,\ie $(-\infty, 0)$ (resp. non-negative half chain, \ie $[0, \infty)$). It is natural to ask if we always have similar factorization for states. The answer is generally no, because there can be quantum entanglement. To make factorization of states more precise, we define the following important split property.
\begin{definition}
    A factor state $\psi$ of $\A^{ql}$ splits or has the split property if there exist states $\psi_{<0}$ and $\psi_{\geqslant0}$ for $\A^{ql}_{<0}$ and $\A^{ql}_{\geqslant0}$ respectively, such that $\psi\sim\psi_{<0}\otimes\psi_{\geqslant0}$ (note it is quasi-equivalence here).
\end{definition}

Intuitively speaking, if a state splits, then there is limited amount of quantum entanglement between the left half and right half of the underlying quantum spin chain.

\begin{remark}
As shown in Corollary A.4 of Ref. \cite{Liu2024LRLSM}, if a state $\psi$ satisfies the split property, \ie $\psi\sim\psi_{<0}\otimes\psi_{\geqslant 0}$, then for any finite integer $n$, $\psi\sim\psi_{<n}\otimes\psi_{\geqslant n}$. Namely, if a state splits at any site, then it splits at all sites.
\end{remark}
For 1d lattice system, it turns out the entanglement area law implies the split property for factor states.
\begin{lemma}[Theorem 1.5 in Ref.~\cite{matsui2011boundedness}]\label{lemma:split}
    For a 1d lattice system, if a factor state $\psi$ satisfies entanglement area law, then it splits.
\end{lemma}
We remark that the original statement of Theorem 1.5 in Ref.~\cite{matsui2011boundedness} does not assume that $\psi$ is a factor state, but its proof actually assumes it.

\subsection{Anti-unitary equivalence}

In this paper, we will disuss both unitary and anti-unitary symmetries. The treatment of unitary symmetries is described before, which is given by a homomorphism from the symmetry group $G$ to the group of LPA, $\G^{lp}$. The consequences of such symmetries are also throughly explored in Ref. \cite{Liu2024LRLSM}. To discuss anti-unitary symmetries, below we first define anti-unitary equivalence of representations (see Appendix \ref{sec:comp_conj} for more details on anti-linear maps and complex conjugations). Our terminology follows Ref.~\cite{Noel2020antiunitary}.

\begin{definition}\label{def:anti-unitary}
    Let $(\cH_{1},\cH_{2})$ be two Hilbert spaces with inner products $\la\cdot,\cdot\ra_{1}$ and $\la\cdot,\cdot\ra_{2}$, respectively. An anti-linear map $\tilde{U}:\cH_{1}\to\cH_{2}$ is called anti-unitary if
    \beq
    \la \tilde{U}a,\tilde{U}b\ra_{2}=\overline{\la a,b\ra_{1}},\forall\,a,b,\in\cH_{1}
    \eeq
    where $\bar{\cdot}$ means complex conjugation. Let $(\pi_{1},\cH_{1})$ and $(\pi_{2},\cH_{2})$ be representations of $\A^{ql}$ and $\bar{\A}^{ql}$ respectively, where $\bar{\A}^{ql}$ is the complex conjugate of $\A^{ql}$ (see Definition \ref{def: conjugate algebra} for its definition). These representations are anti-unitarily equivalent if there is anti-unitary map $\tilde{U}:\cH_{1}\to\cH_{2}$ such that
    \beq
    \tilde{U}\pi_{1}(A)\tilde{U}^{-1}= \pi_{2}(A)
    \eeq
    where $A$ is regarded as an element in $\A^{ql}$ (resp. $\bar{\A}^{ql}$) on LHS (resp. RHS).
\end{definition}
\begin{remark}
    Sometimes, $\cH_{1}$ can be identified with $\cH_{2}$. In this case, $\tilde{U}$ is simply an anti-unitary operator on $\cH_{1}$.
\end{remark}

Now, we state a key lemma, which is a generalization of Theorem 2.2 of Ref.~\cite{Ogata2019TRS}.
\begin{lemma}\label{lemma:Jordan_equiv}
    Let $\psi$ be a factor state of $\A^{ql}$ and $\tilde{\alpha}\in\tA(\A^{ql})$, then the following two conditions are equivalent.
    \begin{enumerate}
        \item There is a quasi-equivalence between $\psi$ and $\psi_{\tilde{\alpha}}$, where $\psi_{\tilde{\alpha}}$ is defined by {\footnote{Note that $\psi\circ\tilde{\alpha}$ is not a state of $\A^{ql}$ if $\tilde{\alpha} $ is anti-linear.}}
        \beq
        \psi_{\tilde{\alpha}}(A):=\begin{cases}
            \psi(\tilde{\alpha}(A)),\quad\text{for linear $\tilde{\alpha}$};\\
            \psi(\tilde{\alpha}(A)^\dagger),\quad\text{for anti-linear $\tilde{\alpha}$.}
        \end{cases}
        \eeq
        \item     For any $\epsilon>0$, there is a finite subset $\Gamma_{\epsilon}$ of the lattice, such that  
        \beq
        |\psi(A)-\psi_{\tilde{\alpha}}(A)|<\epsilon||A||,\quad\forall\,A\in\A^{ql}_{\Gamma^{c}_{\epsilon}}
        \eeq
        where $\Gamma_{\epsilon}^{c}$ means the complement of $\Gamma_{\epsilon}$.
    \end{enumerate}
\end{lemma}

\begin{remark}
In the language of usual quantum mechanics, the condition $\psi=\psi_{\tilde\alpha}$ means that the state $\psi$ is invariant under the operation $\tilde\alpha$.
\end{remark}

\begin{remark}
    If $\psi$ is a pure state with $\psi_{\tilde{\alpha}}=\psi$ and $\tilde{\alpha}$ is of order 2 (\ie $\tilde{\alpha}^{2}=1$), then above lemma reduces to Theorem 2.2 of Ref.~\cite{Ogata2019TRS}. This lemma can also be viewed as an anti-unitary version of Proposition 3.2.8 of Ref.~\cite{Naaijkens_2017}.
\end{remark}
As an important corollary, we have
\begin{corollary}\label{corollary:anti-unitary_map}
    Let $\psi$ and $\psi_{\talpha}$ satisfy either assumptions in Lemma \ref{lemma:Jordan_equiv}. Then there exists an operator $\tilde{U}$ on $\cH_{\psi}$ such that
    \beq
    \tilde{U}\pi_{\psi}(a)\tilde{U}^{-1}= \pi_{\psi}(\talpha(a)),\forall\,a
    \eeq
    where $\tilde{U}$ is unitary (resp. anti-unitary) if $\tilde{\alpha}$ is linear (resp. anti-linear).
    
\end{corollary}

The proofs to Lemma \ref{lemma:Jordan_equiv} and Corollary \ref{corollary:anti-unitary_map} are left to Appendix~\ref{app:Jordan}.

\subsection{Hamiltonians and ground states} \label{subsec: Hamiltonians and ground states}

The above discussions are mostly about operators and states, without referring to any Hamiltonian. In this subsection, we introduce Hamiltonians, ground states and locally unique gapped ground states.

In quantum mechanics with finitely many degrees of freedom, Hamiltonians are operators in the Hilbert space. In infinite-size systems, however, the notion of the total Hilbert space does not make sense, and the Hamiltonians are also not operators. Instead, Hamiltonians are defined by their commutators with operators. Specifically, the Hamiltonians of our interest, which are deembed admissible Hamiltonians, are defined as follows.

\begin{definition}
    In a 1D lattice system $\Lambda=\z$ where sites are labeled by $i$, the Hamiltonian $\delta_{H}$ is admissible if $\delta_{H}=\sum_{Z:|Z|\leqslant k}[h_{Z},\bullet]$ with $h_{Z}$'s (\ie Hamiltonian terms supported on $Z$) satisfying
    \beq\label{eq:admissible_H}
    \begin{split}
        \sup_{i\in\z}||h_{i}||&<B,\\
        \sup_{i\in\z}\sum_{\substack{Z:|Z|\leqslant k,Z\owns i\\\diam(Z)=r}}||h_{Z}||&<\frac{J}{r^{\mathfrak{a}}},\,{\rm with\ }\mathfrak{a}>2,
    \end{split}
    \eeq
    where $\diam(Z):=\max_{x,y\in Z}d(x,y)$, k is an integer, $J$ and $B$ are positive constants and $h_{i}$ is a one-body potential at site $i$.
\end{definition}

In the above, $\delta_H$ gives the commutator between the Hamiltonian and a local operator. The first condition in Eq. \eqref{eq:admissible_H} means that the single-body terms in the Hamiltonian are always bounded, and the second condition means that the interaction strengths decay fast enough when the ranges of the interactions increase. In particular, if the Hamiltonian consists of only 2-body interactions, then the 2-body interaction decays faster than $1/r^2$, where $r$ is the distance between the two interacting spins.

Given the notion of Hamiltonians, now we define ground states in the operator algebra formalism.
\begin{definition} \label{def: ground state}
    A state $\psi$ is said to be the ground state of the Hamiltonian $\delta_{H}$ if 
    \beq\label{eq:ground_states}
    \psi(A^{\dagger}\delta_{H}(A))\geqslant 0,\forall\,A\in\A^{l}
    \eeq
\end{definition}
In the language of ordinary quantum mechanics, this condition means that applying any operator to the ground state cannot decrease the expectation value of the Hamiltonian. In infinite-size systems, the ground states defined above should be viewed as the ``lowest-energy" states within a superselection sector.

A locally unique gapped ground state, which may simply be called a gapped ground state, is defined as follows. Physically, a locally unique gapped ground state is a gapped ground state in a superselection sector.

\begin{definition}\label{definition:locally-unique_gs}
    A ground state $\psi$ of Hamiltonian $\delta_{H}$ is a locally unique gapped ground state if there is a $\gamma>0$  such that
    \beq\label{eq:locally_unique_gapped_gs}
    \psi(A^{\dagger}\delta_{H}A)\geqslant\gamma\psi(A^{\dagger}A)
    \eeq
    for any $A\in \A^{l}$ with $\psi(A)=0$. The energy gap $\Delta$ is the largest possible $\gamma$ satisfying the above inequality.
\end{definition}

Some important properties of locally unique gapped ground states of an admissible Hamiltonian are proved in Ref. \cite{Liu2024LRLSM}. In particular, it is shown there in that locally unique gapped ground states of an admissible Hamiltonian satisfies the area law of entanglement entropy (see Theorem IV.1 therein and Theorem 4.4 of Ref. \cite{Ukai2024}), they are pure (see Theorem A.7 therein) and also satisfy the split property (see Corollary F.1 therein).

\section{Twisted locality-preserving automorphisms, their GNVW index, and symmetry actions} \label{sec: twisted LPA}

After reviewing the operator algebra formalism in Sec. \ref{sec: background and terminology}, in this section, we discuss a twisted version of QCA and LPA to describe anti-unitary symmetries such as time reversal. 

First, we define twisted automorphisms, which are generalizations of the automorphisms in Definition \ref{def: automorphisms} to include anti-linear operations.
\begin{definition}[Twisted automorphisms]
    A map $\alpha:\A^{ql}\to\A^{ql}$ is called a twisted automorphism of $\A^{ql}$ if the followings are true:
    \begin{enumerate}
        \item $\alpha(A+B)=\alpha(A)+\alpha(B)$
        \item $\alpha(AB)=\alpha(A)\alpha(B)$
        \item $\alpha(A^{\dagger})=\alpha(A)^{\dagger}$
        \item $\alpha(\lambda A)=\lambda\alpha(A),$ or $\alpha(\lambda A)=\overline{\lambda}\alpha(A)$, $\forall\,\lambda\in\bbC$, where $\bar{\lambda}$ means the complex conjugate of $\lambda$.
        \item $\alpha$ is invertible.
    \end{enumerate}
    If $\alpha(\lambda A)=\lambda \alpha(A)$, we say that $\alpha$ is linear. Otherwise, it is said to be anti-linear. Twisted automorphisms form a group under finite compositions, denoted by $\tA(\A^{ql})$.
\end{definition}
\begin{remark}
    It is important to note that the $\dagger$ (more precisely, the $*$-operation) on $\A^{ql}$ is not a twisted automorphism because it violates the property (2) above.
\end{remark}

There are multiple ways to characterize twisted automorphisms, which should be familiar from the discussion of time reversal symmetries in quantum mechanics. First, let $\alpha\in\tA(\A^{ql})$ and $\beta\in\Aut(\A^{ql})$. One can verify that $\alpha\beta\alpha^{-1}\in\Aut(\A^{ql})$. As a consequence, $\Aut(\A^{ql})$ is a normal subgroup of $\tA(\A^{ql})$. 
The quotient group $\tA(\A^{ql})/\Aut(\A^{ql})=\z_2$, because $\alpha^{2}\in\Aut(\A^{ql})$ is linear for any $\alpha\in\tA(\A^{ql})$. Therefore, there is a short exact sequence
\beq
1\rightarrow\G^{lp}\stackrel{i}{\hookrightarrow}\tG^{lp}\stackrel{\phi}{\longrightarrow}\z_{2}\rightarrow 1,
\eeq
where $i$ denotes the inclusion map and $\phi$ denotes the quotient map. For $\tilde\alpha\in\tA(\A^{ql})$, it is linear if $\phi(\tilde\alpha)=0$ and anti-linear if $\phi(\tilde\alpha)=1$.

In fact, $\tA(\A^{ql})$ is a semi-direct product, \ie
\beq
\tA(\A^{ql}) \simeq \Aut(\A^{ql})\rtimes  \z_{2}
\eeq
To see it, note that for each site $n\in\z$ one can choose a complex conjugation operator $\Theta_{n}$, which is an anti-linear operator satisfying $\Theta_{n}^{2}=1$. Writing $K_{n}:=\Ad_{\Theta_{n}}$, one can formally define (see Appendix \ref{sec:comp_conj} for more detail)
\beq\label{eq:conjugation_operator}
K:=\prod_{n\in\z}K_{n}
\eeq
Obviously, $K\in \tA(\A^{ql})$ and $K^{2}=1$, so it generates a $\z_{2}$ subgroup of $\tA(\A^{ql})$. Note that $\alpha:=\tilde{\alpha} K$ is a linear automorphism. Thus, any element $\tilde{\alpha}\in\tA(\A^{ql})$ can be written as
\beq \label{eq: decomposing anti-linear map into linear and K}
\tilde{\alpha} = \alpha K
\eeq
for some $\alpha\in\Aut(\A^{ql})$ and $K$ defined in Eq.~\eqref{eq:conjugation_operator}. We stress that $K$ defined above is not unique, because $\Theta_n$ is not unique. However, we will show that all our main results are valid for all choices of $K$.

Following the last section, we define twisted locality-preserving automorphisms as follows.
\begin{definition}
    Let $\alpha\in\tA(\A^{ql})$. It is called a twisted locality-preserving automorphism if for any local operator $A\in\A_{X}$ and $r>0$, there exists a local operator $B\in\A_{B(X,r)}$ and a non-negative decreasing function $f_{\alpha}(r)\searrow 0$, such that
    \beq
    ||\alpha(A)-B||<f_{\alpha}(r)||A||
    \eeq
    Especially, $f_{\alpha}$ is the same for all $A\in\A_X$. The group of twisted LPA is denoted by $\tG^{lp}$. If there exists a positive number $R$ such that $f_{\alpha}(r)=0,\,\forall\,r>R$, then we say that $\alpha$ is a twisted quantum cellular automaton. The smallest possible choice of $R$ is called the radius of the twisted QCA $\alpha$, denoted by $r_{\alpha}$. The group of twisted QCA is denoted by $\tG^{\QCA}$.
\end{definition}

From the above definition, it is readily verified that $K$ defined in Eq.~\eqref{eq:conjugation_operator} is a twisted QCA and therefore a twisted LPA.

Following a similar discussion as that on $\tA(\A^{ql})$, we have the following semi-direct product decomposition
\beq\label{eq:short_ex_seq}
\begin{split}
    \tG^{lp}&\simeq \G^{lp}\rtimes \z_{2} \\
    \tG^{\QCA}&\simeq \G^{\QCA}\rtimes \z_{2}
\end{split}
\eeq
Note that the anti-unitary map $K\in\tG^{\QCA}\subset\tG^{lp}$ and $K^{2}=1$, so it generates above $\z_{2}$ factor.

One can define the GNVW index for an element in $\tG^{lp}$ as follows. If $\tilde{\alpha}\in\G^{lp}$ (\ie it is linear), then $\ind(\tilde{\alpha})$ is already defined in Eq.~\eqref{eq:GNVW_index}. On the other hand, if $\tilde{\alpha}$ is anti-linear, according to the above discussion, we can write $\tilde{\alpha}=\alpha K$. Then we define the GNVW index of $\tilde\alpha$ as $\ind(\tilde{\alpha}):=\ind(\alpha)$. It is important to show that this index is a well-defined homomorphism $\ind:\tG^{lp}\to\z[\{\log p_{j}\}_{j\in J}]$. In particular, because there can be different $\alpha$'s and $K$'s such that $\tilde\alpha=\alpha K$, we must show that this index does not depend on how we decompose $\tilde\alpha$ into $\alpha K$. The proof that this index is well-defined will be presented in Appendix~\ref{app:index}.

With the above preparations, we can now discuss symmetry actions for both unitary and anti-unitary symmetries. Suppose that the symmetry group is $G$. By a $G$ symmetry action, we mean a group homomorphism
$\tilde{\alpha}:G\to\tG^{lp}$. Note that there is an induced map 
\beq\label{eq:twist_indicator}
\varphi:G\stackrel{\tilde{\alpha}}{\longrightarrow}\tG^{lp}\stackrel{\phi}{\longrightarrow}\z_{2}.
\eeq
For a group element $g\in\G$, we say that it is unitary if $\varphi(g)=0$ and anti-unitary if $\varphi(g)=1$. Notice, strictly speaking, whether the symmetry is unitary or anti-unitary does not only depend on $g$ itself, but also on the symmetry action $\tilde{\alpha}$. 

We end this section by discussing how the multiplication and associativity work for the symmetry actions described by twisted LPA. Since $\tG^{lp}=\G^{lp}\rtimes \z_{2}$, for any $g\in G$ we have
\beq\label{eq:tLPA_decomposition}
\tilde{\alpha}(g)=\alpha(g)\theta(g)
\eeq
where $\alpha(g)\in\G^{lp}$ and
\beq\label{eq:theta}
\theta(g):=\begin{cases}
    1,\varphi(g)=0\\
    K,\varphi(g)=1
\end{cases}
\eeq
By assumption, as a homomorphism, $\tilde{\alpha}$ satisfies $\tilde{\alpha}(g)\tilde{\alpha}(h)=\tilde{\alpha}(gh)$, which gives the following composition rules for $\alpha(g)$ and $\theta(g)$:
\beq\label{eq:homomoprhism}
\begin{split}
    \alpha(g)(\theta(g)\triangleright\alpha(h))&=\alpha(gh)\\
    \theta(g)\theta(h)&=\theta(gh)
\end{split}
\eeq
where $\theta(g)\triangleright \alpha(h):=\theta(g)\alpha(h)\theta(g)^{-1}$.

\section{Anomaly index}\label{sec:construction_anomaly_index}

After introducing twisted LPA, we can now define an anomaly index for a symmetry action $\tilde{\alpha}:G\to\tG^{lp}$. If this anomaly index takes a nontrivial value (in the sense to be introduced), then we say that the symmetry action is anomalous. Otherwise, the symmetry action is said to be non-anomalous or anomaly free. In essence, the anomaly characterizes the interplay between locality and the symmetry. Later, we will formulate various LSM-type constraints due to an anomalous symmetry.

\subsection{General definition}

Our construction of the anomaly index is inspired by Ref.~\cite{kapustin2024anomalous}. Ref. \cite{kapustin2024anomalous} focuses on unitary symmetries described by LPA, and we extend the notion of the anomaly index to both unitary and anti-unitary symmetries described by general twisted LPA.

To construct the anomaly index, we define several new objects:
\begin{enumerate}

    \item $\A^{ql}_{<0}$: Algebra of quasi-local operators supported on negative half-chain, \ie $(-\infty, -1]$.
    
    \item  $\A^{ql}_{\geqslant 0}$: Algebra of quasi-local operators supported on non-negative half-chain, \ie $[0, \infty)$.
    
    \item $\G^{lp}_{0}$: The subgroup of $\tG^{lp}$ generated by elements of the form of $\Ad_{U}$, where $U$ is a quasi-local unitary operators.
    
    \item $\tG^{lp}_{<0}$: Twisted locality-preserving automorphisms of $\A^{ql}_{<0}$ only, acting trivially on $\A^{ql}_{\geqslant0}$.
    
    \item $\tG^{lp}_{\geqslant 0}$: Twisted locality-preserving automorphisms of $\A^{ql}_{\geqslant0}$ only, acting trivially on $\A^{ql}_{<0}$.
    
    \item  $\tG^{lp}_{+}$ (resp. $\tG^{lp}_{-}$) is the group generated by $\tG^{lp}_{\geqslant 0}\G^{lp}_{0}$ (resp. $\tG^{lp}_{<0}\G^{lp}_{0}$).
\end{enumerate}

Note that we always have 
\beq\label{eq:algebra_decomp}
\A^{ql}\simeq\A^{ql}_{<0}\otimes\A^{ql}_{\geqslant 0}
\eeq
Moreover, $\tG^{lp}$ has a natural group action on $\G_{0}$ by conjugation. Concretely, for any $\tilde{\alpha}\in\tG^{lp}$ and $U\in\cU^{ql}$, we have
\beq
\tilde{\alpha}\triangleright \Ad_{U}:=\tilde{\alpha}\circ\Ad_{U}\circ\tilde{\alpha}^{-1}=\Ad_{\tilde{\alpha}(U)}\in \G^{lp}_{0}
\eeq

One must be careful that $\tG^{lp}_{<0}$ and $\tG^{lp}_{\geqslant 0}$ are not a subgroup of $\tG^{lp}$. One might think that given an anti-linear $\tilde{\alpha}_{\geqslant0}\in\tG^{lp}_{\geqslant0}$, one can always obtain a twisted automorphism of $\A^{ql}\simeq \A^{ql}_{<0}\otimes\A^{ql}_{\geqslant0}$ by considering $\id_{\A^{ql}_{<0}}\otimes\tilde{\alpha}_{\geqslant0}$. However, $\id_{\A^{ql}_{<0}}\otimes\tilde{\alpha}_{\geqslant0}$ is neither linear nor anti-linear on $\A^{ql}$ (see Appendix \ref{sec:comp_conj} for more details), hence it is not an element of $\tG^{lp}$. However, let $\tilde{\alpha}_{<0}\in\tG^{lp}_{<0}$ be another anti-linear element, one can show that $\tilde{\alpha}_{<0}\otimes\tilde{\alpha}_{\geqslant 0}$ is anti-linear on $\A^{ql}$ and thus an element of $\tG^{lp}$.

Below we give the general definition of the anomaly index, which relies on the following two propositions that are proved in Appendix \ref{sec:decomposition}.
\begin{proposition}\label{prop:decomp_twisted}
    Let $\tilde{\alpha}\in\tG^{lp}$. Then $\tilde\alpha$ has a vanishing GNVW index if and only if it admits a decomposition
    \beq
    \tilde{\alpha}=(\tilde{\alpha}_{<0}\otimes \id)\alpha_{0}(\id\otimes\tilde{\alpha}_{\geqslant0})
    \eeq
    where $\tilde{\alpha}_{<0}\in\tG^{lp}_{<0}$, $\alpha_{0}\in \G_{0}^{lp}$ and $\tilde{\alpha}_{\geqslant0}\in\tG^{lp}_{\geqslant0}$. Moreover, $\tilde{\alpha}_{\geqslant 0}$ and $\tilde{\alpha}_{<0}$ are anti-linear on $\A^{ql}_{\geqslant 0}$ and $\A^{ql}_{<0}$, respectively, if and only if $\tilde{\alpha}$ is anti-linear. For notational simplicity, below we write $\id\otimes\tilde{\alpha}_{\geqslant0}$ as $\tilde{\alpha}_{\geqslant0}$ and $\tilde\alpha_{<0}\otimes\id$ as $\tilde\alpha_{<0}$.
\end{proposition}

For $\tilde\alpha\in\tG^{lp}$ with a vanishing GNVW index, such a decomposition is not unique. The following result characterizes the relation between different decompositions.
\begin{proposition}\label{prop:nonunique_decomp}
    Let $\tilde{\alpha}\in\tG^{lp}$ with a vanishing GNVW index and suppose that it decomposes as
    \beq
    \tilde{\alpha}_{<0}\alpha_{0}\tilde{\alpha}_{\geqslant0}=\tilde{\alpha}=\tilde{\beta}_{<0}\beta_{0}\tilde{\beta}_{\geqslant0}.
    \eeq
    Then there exists a quasi-local unitary $V$ such that
    \beq    \tilde{\beta}_{\geqslant0}=\Ad_{V}\circ\tilde{\alpha}_{\geqslant0}.
    \eeq
    Especially, $\Ad_{V}$ is only supported on the non-negative half chain.
\end{proposition}

With these results, one can construct the anomaly index as follows. 

We first consider a simpler case, where $\tilde{\alpha}:G\to\tG^{lp}$ is a symmetry action with $\ind(\tilde{\alpha}(g))=0$ for all $g\in G$. The more general case will be discussed later. By Proposition \ref{prop:decomp_twisted}, we have
\beq \label{eq: decomposing twisted LPA}
\tilde{\alpha}(g)=\tilde{\alpha}_{<0}(g)\alpha_{0}(g)\tilde{\alpha}_{\geqslant 0}(g)
\eeq
Since $\tilde{\alpha}$ is a group homomorphism, we have $\tilde{\alpha}(g)\tilde{\alpha}(h)=\tilde{\alpha}(gh)$. So one obtains two different decompositions of $\tilde{\alpha}(gh)$:
\beq
\tilde{\alpha}_{<0}(g)\tilde{\alpha}_{<0}(h)\beta_{0}\tilde{\alpha}_{\geqslant0}(g)\tilde{\alpha}_{\geqslant0}(h)=\tilde{\alpha}_{<0}(gh)\alpha_{0}(gh)\tilde{\alpha}_{\geqslant0}(gh)
\eeq
where $\beta_{0}:=(\tilde{\alpha}_{<0}^{-1}(h)\triangleright\alpha_{0}(g))(\tilde{\alpha}_{\geqslant0}(g)\triangleright\alpha_{0}(h))\in\G^{lp}_{0}$. Then according to Proposition \ref{prop:nonunique_decomp}, there exists a quasi-local unitary $V(g,h)$ such that
\beq\label{eq:near_homomoprhism2}
\tilde{\alpha}_{\geqslant0}(g)\tilde{\alpha}_{\geqslant0}(h)=\Ad_{V(g,h)}\tilde{\alpha}_{\geqslant0}(gh)
\eeq

Next, one uses the fact that associativity demands that $(\tilde{\alpha}_{\geqslant0}(g)\tilde{\alpha}_{\geqslant0}(h))\tilde{\alpha}_{\geqslant0}(k)=\tilde{\alpha}_{\geqslant0}(g)(\tilde{\alpha}_{\geqslant0}(h)\tilde{\alpha}_{\geqslant0}(k))$, from which one deduces that
\beq
\Ad_{\omega(g,h,k)}=1
\eeq
where 
\beq\label{eq:anomaly_index}
 \omega(g,h,k):=V(g,h)V(gh,k)V(g,hk)^{-1}(\tilde{\alpha}_{\geqslant0}(g)\triangleright V(h,k))^{-1}
\eeq
So $\omega(g,h,k)\in\U$.

Moreover, by a brute force calculation (see Appendix \ref{subapp: 3-cocycle equation}), one sees that $\omega(g, h, k)$ satisfies a twisted 3-cocycle equation:
\beq \label{eq: 3-cocycle equation main}
(\tilde{\delta}^{(4)}\omega)(g, h, k, l)=\frac{(\theta(g)\triangleright\omega(h, k, l))\omega(g, h, k)\omega(g, hk, l)}{\omega(gh, k, l)\omega(g, h, kl)}=1,
\eeq
where $\theta(g)\triangleright\omega(h, k, l)=\omega(h, k, l)$ if $g$ is unitary, and $\theta(g)\triangleright\omega(h, k, l)=\overline{\omega(h, k, l)}$ if $g$ is anti-unitary.

Also, note that $V(g, h)$ in Eq. \eqref{eq:near_homomoprhism2} has an instrinsic phase ambiguity, because $\Ad_{V(g, h)}=\Ad_{V(g, h)/\lambda(g, h)}$ for any $\lambda(g, h)\in\U$. By replacing $V(g, h)$ with $V(g, h)/\lambda(g, h)$, according to Eq. \eqref{eq:anomaly_index}, $\omega(g, h, k)$ changes according to
\beq
\omega(g,h,k)\rightarrow\omega(g, h, k)(\tilde{\delta}^{(3)}\lambda)(g, h, k)
\eeq
with
\beq \label{eq: defining coboundary}
(\tilde{\delta}^{(3)}\lambda)(g, h, k)=\frac{(\theta(g)\triangleright \lambda(h, k))\lambda(g, hk)}{\lambda(g, h)\lambda(gh, k)},
\eeq

Therefore, one can define a twisted cohomology group $\rH^3_\varphi(G; \U)$:
\beq
\rH^3_\varphi(G; \U)=\ker(\tilde{\delta}^{(4)})/\im(\tilde{\delta}^{(3)}),
\eeq
where a unitary $g\in G$ does not act on the $\U$ coefficient, while an anti-unitary $g\in G$ acts on the coefficient $\U$ by complex conjugation.

Now we can define the anomaly index.

\begin{definition}
    Given a $G$ symmetry action $\tilde\alpha: G\rightarrow\tG^{lp}$ where $\tilde\alpha(g)$ has a vanishing GNVW index for all $g\in G$, its anomaly index is defined as the element in $\rH^3_\varphi(G; \U)$, which $\omega(g, h, k)$ in Eq. \eqref{eq:anomaly_index} corresponds to.
\end{definition}

The validity of this definition is ensured by the following theorem, which is proved in Appendix \ref{app:twisted_cocycle}.
\begin{theorem}\label{theorem:anomaly_index}
    The phase factor $\omega:G^{3}\to\U$ defined in Eq. \eqref{eq:anomaly_index} is a twisted 3-cocycle, and its corresponding cohomology class in $\rH^{3}_{\varphi}(G;\U)$ is independent of the choice of decomposition $\tilde{\alpha}_{\geqslant0}$ in Eq. \eqref{eq: decomposing twisted LPA} and the site at which the decomposition is carried out.
\end{theorem}

Next, we turn to the more general case where some symmetry actions may be described by a twisted LPA with a nonzero GNVW index. Concretely, suppose that the image of symmetry action $\tilde{\alpha}:G\to\tG^{lp}$ has nontrivial GNVW index, then $\tilde\alpha$ and Eq. \eqref{eq:short_ex_seq} induce the following map
\beq
\tau:G\stackrel{\tilde{\alpha}}{\longrightarrow} \tG^{lp}\longrightarrow\text{generalized translation}
\eeq
So one can stack another copy of the system with the original system, on which $G$ acts as $\tau(g)^{-1}$. Thus, on the composite system, the symmetry acts as $\tilde{\alpha}(g)\otimes\tau(g)^{-1}$, which has a vanishing GNVW index. This reduces the current situation to the earlier one, and the anomaly index of $\tilde{\alpha}$ is defined to be the anomaly index of $\tilde{\alpha}\otimes\tau^{-1}$.

Note that $\tau(g)^{-1}$ above is linear if $g$ is a unitary symmetry, and $\tau(g)^{-1}$ is anti-linear if $g$ is an anti-unitary symmetry. One might want to demand that $\tau(g)^{-1}$ should always be the ordinary linear translation operation, even if $g$ is anti-unitary. However, if one does this, when $g$ is an anti-unitary symmetry, $\tilde\alpha(g)\otimes\tau(g)^{-1}$ is neither linear nor anti-linear on the composite operator algebra of the composite system, and the discussion above does not apply anymore. Therefore, we always demand that $\tau(g)^{-1}$ should be anti-linear (resp. linear) if $g$ is an anti-unitary (resp. unitary) symmetry. More precisely, denoting by $\tau_0(g)$ the ordinary linear translation operation, then $\tau(g)=\tau_0(g)\theta(g)$, with $\theta(g)$ given in Eq. \eqref{eq:theta}.

From the above definition of the anomaly index, one can see that, roughly speaking, the anomaly characterizes the interplay between the local operator algebra and the symmetry action. For example, if the symmetry acts in a completely local, on-site fashion, then we can take $V(g, h)=1$ in Eq. \eqref{eq:near_homomoprhism2}. According to Eq. \eqref{eq:anomaly_index}, $\omega(g, h, k)=1$. Therefore, an on-site symmetry is anomaly free.

\subsection{Example: $G=G_0\times\z$} \label{subsec: example}

Now we illustrate the concept of the anomaly index in a very important example. In this example, we compute the mixed anomaly between the lattice translation symmetry and an on-site symmetry $G_{0}$ (which can include anti-unitary symmetry actions), \ie $G=G_{0}\times\z$. Here the lattice translation symmetry is the ordinary one, which is unitary, but $G_0$ can contain unitary and anti-unitary symmetries. The calculation is very similar to Example 3.2 of Ref.~\cite{kapustin2024anomalous}, except that Ref.~\cite{kapustin2024anomalous} focuses on unitary symmetries but we treat unitary and anti-unitary symmetries on equal footing. The result is that the anomaly of the $G$ symmetry is in one-to-one correspondence with the class of projective representation under $G_0$, which the degrees of freedom in each unit cell carry. If these degrees of freedom carry a linear representation under $G_0$, then the $G$ symmetry is anomaly free.

Concretely, for any $g\in G_0$ and $n\in\z$, the symmetry action under consideration is formally given by
\beq
\tilde{\alpha}(g,n)=\tau_0^{n}\prod_{j\in\z}\Ad_{\rho_{j}^{-1}(g)}\theta(g)
\eeq
where $\tau_0$ is the generator of translation, $\theta(g)$ is given in Eq. \eqref{eq:theta}, and $\rho_{j}$ is the on-site action of $G_0$ at site $j$ with a projective phase $\eta\in \rH_{\varphi}^{2}(G_{0};\U)$, \ie $\rho_j(g)\left(\theta_j(g)\triangleright\rho_j(g')\right)=\eta(g, g')\rho_j(gg')$ for $g, g'\in G_0$, where $\theta_j(g)=1$ if $g$ is unitary, and $\theta_j(g)=K_j$ if $g$ is anti-unitary.{\footnote{At the $j$-th site, the $G_0$ symmetry maps the state $|\psi_j\rangle$ to $\rho_j(g)\Theta_j^{\varphi(g)}|\psi_j\rangle$, and its action on an operator $O_j$ at this site is $O_j\rightarrow \Ad_{\rho_j^{-1}(g)}\theta_g(O_j)$. Note the inverse in the adjoint action of the symmetry on the operator.}} We assume that $\rho_{j}$'s are the same for all $j\in\z$, otherwise it is not compatible with the translation symmetry.

Physically, $\tilde\alpha(g, n)$ is a composition of translating by $n$ unit cells and an on-site symmetry action, where the degrees of freedom in each unit cell form a projective representation under $G_0$ that is described by $\eta$. For example, when $G_0$ is the usual order-2 time reversal symmetry, \ie $G_{0}=\z_{2}^{T}$, we have $\rH^{2}_{\varphi}(\z_{2}^{T},\U)=\z_{2}$. Denoting the generator of $\z_2^T$ by $T$ and the identity element by $e$, the two elements in $\rH^{2}_{\varphi}(\z_{2},\U)$ are characterized by $\eta(T, T)\eta(e, e)=\pm 1$. The case with $\eta(T, T)\eta(e, e)=1$ (resp. $\eta(T, T)\eta(e,e)=-1$) corresponds to the case where the on-site degrees of freedom are Kramers singlets (resp. Kramers doublets). 
    
The anomaly index takes value in $\rH^{3}_{\varphi}(G_{0}\times\z;\U)$. By the K\"{u}nneth formula (see, \eg Chapter 3 of Ref. \cite{hatcher2002algebraic}),
\beq
\rH^{3}_{\varphi}(G_{0}\times\z;\U_{\varphi})=\rH^{3}_{\varphi}(G_{0};\U)\oplus\rH^{2}_{\varphi}(G_{0};\U).
\eeq
The first part, $\rH^{3}_{\varphi}(G_{0};\U)$, characterizes the anomaly of the $G_0$ symmetry, since it originates purely from $G_0$ and is independent of the translation symmetry. Because $G_0$ is on-site, this part must vanish for the symmetry action under consideration, and the classification of the anomalies in the present case is given by the second part, $\rH^{2}_{\varphi}(G_{0};\U)$. For this part, we note that, in the the K\"{u}nneth formula,  the map from $\rH^{3}_{\varphi}(G_{0}\times\z;\U)\to \rH^{2}_{\varphi}(G_{0};\U)$ is given by the slant product (also known as the loop transgression in the mathematical literature). Explicitly, the slant product maps a twisted 3-cocycle $\omega((g,n), (g',n'), (g'',n''))$ to a 2-cocycle
    \beq\label{eq:slant}
    (\omega/[1])(g, g')=\frac{\omega((e,1), (g,0), (g',0))\omega((g,0), (g',0), (e,1))}{\omega((g,0), (e,1), (g',0))}
    \eeq
    where $\omega/[1]$ denotes the 2-cocycle resulted from the map, and $e\in G_{0}$ is the identity element in $G_0$. One can verify that $\omega/[1]$ indeed satisfies the twisted 2-cocycle equation:
    \beq
    \left(\tilde\delta^{(3)}(\omega/[1])\right)(g, h, k)=\frac{(g\triangleright(\omega/[1])(h, k))(\omega/[1])(g, hk)}{(\omega/[1])(g, h)(\omega/[1])(gh, k)}=1.
    \eeq
    In general, the slant product is just a homomorphism from $\rH^{3}_{\varphi}(G_{0}\times\z;\U)$ to $\rH^{2}_{\varphi}(G_{0};\U)$, but not necessarily a bijective map. In the present case, because $G_0$ is on-site, we will see that it becomes a bijective map in the sense below.
    
    \begin{figure}[h!]
        \centering
        \includegraphics[width=0.5\linewidth]{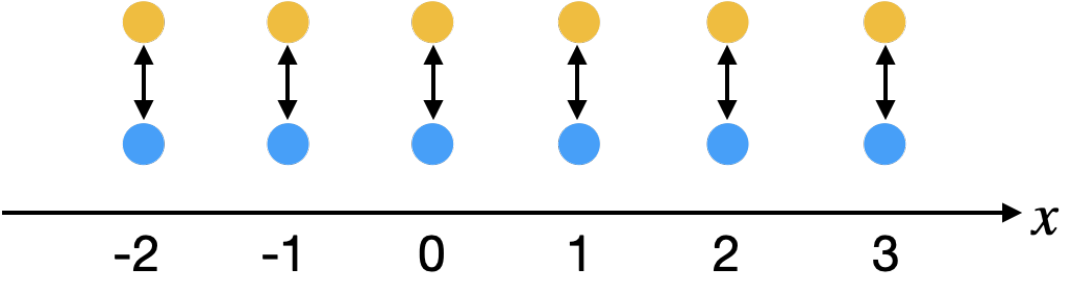}
        \caption{The block-partitioned unitary $S$ which swaps the original system (yellow) and its copy (blue). This figure only shows 6 sites around the origin, but the lattice actually extends from $-\infty$ to $\infty$.}
        \label{fig:swap}
    \end{figure}
      \begin{figure}[h!]
        \centering
        \includegraphics[width=0.5\linewidth]{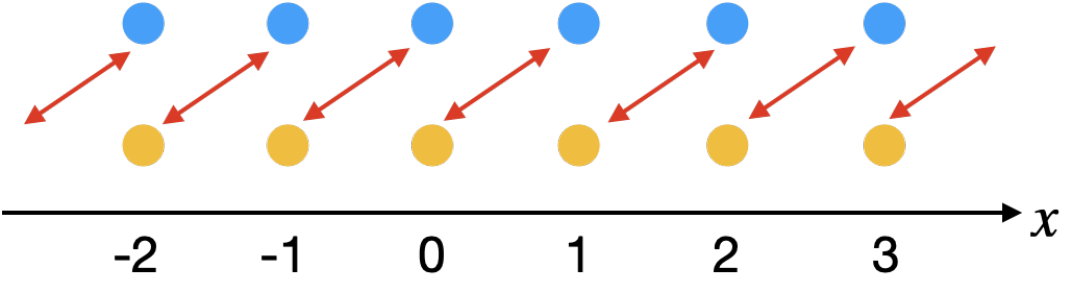}
        \caption{The swap operator with a shift. This figure only shows 6 sites around the origin, but the lattice actually extends from $-\infty$ to $\infty$.}
        \label{fig:shifted_swap}
    \end{figure}

    With the above understanding of the structure of $\rH^{3}_{\varphi}(G_{0}\times\z;\U_{\varphi})$, now we calculate the anomaly index of the symmetry action. More precisely, we will calculate the slant product of the anomaly index, and show that for each type of projective representation characterized by $\eta\in\rH^2_\varphi(G_0; \U)$, the slant product of the anomaly index is precisely $\eta$. Therefore, there is a one-to-one correspondence between the anomalies in this case and the projective representation classes carried by the degrees of freedom in each unit cell. In particular, if these degrees of freedom carry a linear representation, the symmetry action is anomaly free.
    
    Since the symmetry action involves translation $\tau_0$ and thus has a non-vanishing GNVW index, to compute the anomaly 3-cocycle, one needs to stack another copy of the chain on which $G$ acts as $\tau^{-1}$, where $\tau(g, n)=\tau_0(n)\theta(g)$. More explicitly, the generator of translation symmetry is $\tau_0\otimes\tau^{-1}$ on the composite system, and we can write $\tau_0\otimes\tau_0^{-1}=\tilde{S}S$, where $S$ is the swap of first and second copies (see Fig.~\ref{fig:swap}), and $\tilde{S}$ first swaps the two copies and then shifts one copy to the right by one unit cell (see Fig.~\ref{fig:shifted_swap}).

    Therefore, $(g,n)\in G_{0}\times\z$ acts on the composite system as a circuit,
    \beq
    \tilde{\alpha}'(g,n)=\tilde\alpha(g, n)\otimes\tau^{-1}(g, n)=(\tilde{S}S)^{n}\circ(\prod_{j\in\z}\Ad_{\rho_{j}^{-1}(g)}\theta(g)\otimes\theta(g))
    \eeq
    where $\Ad_{\rho_{j}^{-1}(g)}\theta(g)\otimes\theta(g)$ means that the internal symmetry $G_0$ acts on the original system in the original way, given by $\Ad_{\rho_{j}^{-1}(g)}\theta(g)$, while it acts on the stacked copy by $\theta(g)$.

    \begin{figure}
        \centering
        \includegraphics[width=0.7\linewidth]{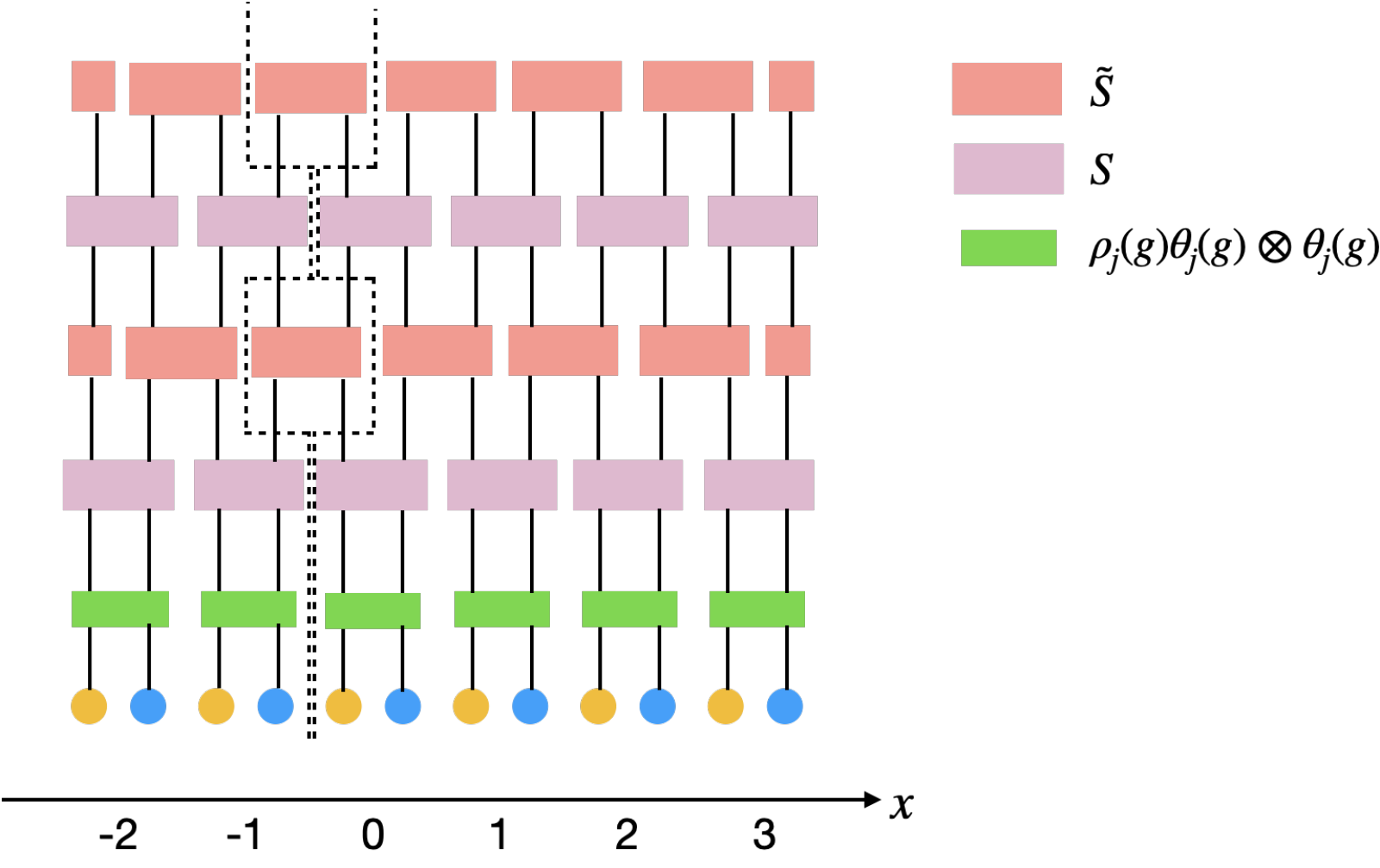}
        \caption{The dashed lines, which separate the figure into 3 parts, perform a decomposition of $\tilde\alpha'(g, n)$ as in Proposition \ref{prop:decomp_twisted}, \ie $\tilde\alpha'_{\geqslant 0}(g, n)$ is the right most part, $\tilde\alpha'(g, n)_{<0}$ is the left most part, and $\alpha'_0(g, n)$ can be obtained from the region between the two dashed lines, $\tilde\alpha'_{\geqslant 0}(g, n)$ and $\tilde\alpha'(g, n)_{<0}$. This figure only shows 6 sites around the origin, but the lattice actually extends from $-\infty$ to $\infty$. Also, $n=2$ in this figure, but the decomposition of $\tilde\alpha'(g, n)$ for other values of $n$ can be similarly obtained.}
        \label{fig: decomposition}
    \end{figure}
    
    Now it is easy to perform the decomposition in Proposition \ref{prop:decomp_twisted}. As in Fig. \ref{fig: decomposition}, we choose
    \beq \label{eq: decomposition example}
    &\tilde{\alpha}'_{\geqslant0}(g,n)=(\tilde{S}_{+}S_{+})^{n}\circ(\prod_{j=1}^{\infty}\Ad_{\rho_{j}^{-1}(g)}\theta_j(g)\otimes\theta_j(g)),
    \eeq
    where $S_{+}$ and $\tilde{S}_{+}$ are the restriction of swaps on the non-negative axis (see Fig.~\ref{fig:partial_swap}). The expressions of $\tilde\alpha'_0(g, n)$ and $\tilde\alpha'_{<0}(g, n)$ can be similarly read off, but they will not be used to calculate the anomaly index.
        \begin{figure}[h!]
        \centering
        \includegraphics[width=0.5\linewidth]{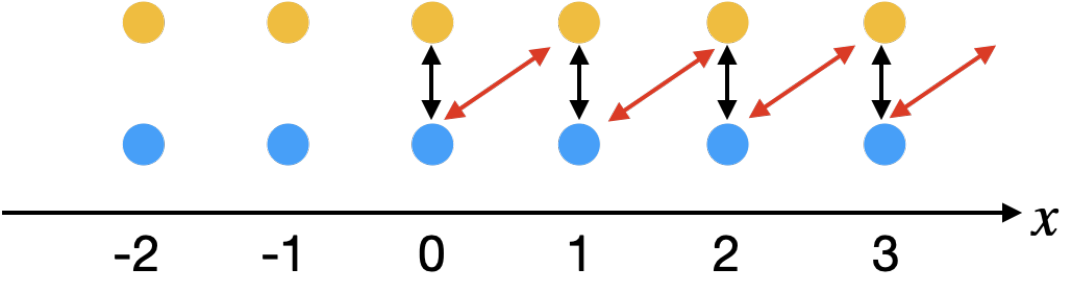}
        \caption{The circuit $\tilde{S}_{+}S_{+}$. The black arrows are implemented first. This figure only shows 6 sites around the origin, but the lattice actually extends from $-\infty$ to $\infty$.}
        \label{fig:partial_swap}
    \end{figure}
    
    Now we can calculate the slant product of the anomaly index using Eqs. \eqref{eq:near_homomoprhism2}, \eqref{eq:anomaly_index} and \eqref{eq:slant}. From Eqs. \eqref{eq:anomaly_index} and \eqref{eq:slant}, we only need $V((g, 0), (g', 0))$, $V((g, 1), (g', 0))$ and $V((g, 0), (g', 1))$ to calculate the slant product. For the choice of $\tilde\alpha'_{\geqslant0}$ in Eq. \eqref{eq: decomposition example}, we can check that $V((g,0), (g',0))=V((g,1), (g',0))=1$ and $V((g,0), (g',1))=\left(\rho_{0}^{-1}(g)\right)\otimes\id$ are compatible with Eq. \eqref{eq:near_homomoprhism2}. Then using Eq. \eqref{eq:anomaly_index} and Eq.~\eqref{eq:slant}, we obtain
    \beq \label{eq: slant product result}
    (\omega/[1])(g, g')=\left(\rho_0(gg')\right)^{-1}\rho_0(g)\left(g\triangleright\rho_0(g')\right)=\eta(g, g').
    \eeq
    This establishes the aforementioned one-to-one correspondence between the anomalies of the symmetry actions under consideration and the projective representation class of the degrees of freedom in each unit cell.

    Before ending this subsection, we note that the above result of the slant product of the anomaly index allows us to write down the anomaly index itself. The K\"{u}nneth formula for cohomology rings shows that any $\omega\in\rH_{\varphi}^{3}(G_{0}\times \z;\U)$ can be rewritten as (see Chapter 3 of Ref. \cite{hatcher2002algebraic})
    $$\label{eq:decomposition_anomaly}
    \omega = \tilde{\omega}+\eta\cup x
    $$
    for some uniquely determined $\tilde{\omega}\in\rH^{3}_{\varphi}(G_{0};\U),\eta\in\rH^{2}_{\varphi}(G_{0};\U)$ and $x\in\rH^{1}(\z;\U)$.
    Let $i:G_{0}\hookrightarrow G_{0}\times\z$ be the natural inclusion, we have $i^{*}\omega =\tilde{\omega}$ which characterizes the anomaly solely from $G_{0}$ itself, and it vanishes if $G_{0}$ is on-site. The slant product is nothing but a contraction map along $x$, as is easily checked that $(\eta\cup x)/[1]=\eta$. Therefore, the anomaly 3-cocycle can be written as
    \beq \label{eq: anomaly 3-cocycle result}
    \omega((g_1, n_1), (g_2, n_2), (g_3, n_3))=\eta(g_1, g_2)^{n_3}.
    \eeq
    One can easily see that Eq. \eqref{eq: anomaly 3-cocycle result} and Eq. \eqref{eq:slant} together lead to Eq. \eqref{eq: slant product result}, which shows the consistency of these rsults. Eq. \eqref{eq: anomaly 3-cocycle result} will be used in Sec. \ref{subsec: anti-unitary translation}.

\section{Generalized Lieb-Schultz-Mattis theorems for infinite systems with anti-unitary symmetry} \label{sec: LSM infinite chains}

After defining the anomaly of a symmetry in Sec. \ref{sec:construction_anomaly_index}, in this and the next sections, we explore the consequences of an anomalous symmetry. In this section, we focus on the consequences on infinite-size systems, where notions like locality are most cleanly defined. In the next section, we discuss the consequences of an anomalous symmetry on large but finite systems, which are physically relevant.

\subsection{The root lemma and its corollaries}

All our theorems are rooted in the following lemma, which can be viewed as a generalized Lieb-Schultz-Mattis (LSM) theorem.

\begin{lemma} \label{lemma: root lemma}
Let $\psi$ be a type-I factor state satisfying the split property, and $\tilde{\alpha}:G\to \tG^{lp}$ be a $G$-symmetry action with anomaly index $\omega\in\rH_\varphi^3(G; \U)$. If $\psi$ is symmetric under $G$-symmetry, \ie $\psi=\psi_{\tilde\alpha(g)}$ for all $g\in G$, then $\omega$ vanishes.
\end{lemma}

The proof of this root lemma is presented in Appendix \ref{app: proving the root lemma}. Here we derive some important corollaries from this lemma.

\begin{corollary} \label{corollary: cluster + area law}
    Let $\psi$ be a (not necessarily pure) state satisfying the clustering property and the area law of entanglement entropy, and $\tilde\alpha: G\to\tG^{lp}$ be a $G$ symmetry action with anomaly index $\omega\in\rH^3_\varphi(G; \U)$. If $\psi$ is symmetric under $G$, \ie $\psi=\psi_{\tilde\alpha(g)}$ for all $g\in G$, then $\omega$ vanishes.
\end{corollary}
This corollary is a powerful and elegant theorem that concerns the interplay between correlation, entanglement and symmetry.

\begin{proof}[Proof of Corollary \ref{corollary: cluster + area law}] Because $\psi$ satisfies the clustering property, according to Theorem \ref{thm:clustering}, $\psi$ is a factor state. Because $\psi$ satisfies the entanglement area law, according to Lemma \ref{lemma: area law implies type I} and Lemma \ref{lemma:split}, it is a type-I factor state that satisfies the split property. Then by the root lemma (Lemma \ref{lemma: root lemma}), $\omega$ vanishes.
    
\end{proof}

Another corollary from the root lemma is as follows.

\begin{corollary} \label{corollary: Kapustin-Sopenko}
    Let $\tilde\alpha:G\to \tG^{lp}$ be a $G$ symmetry action with anomaly index $\omega\in\rH^3_\varphi(G; \U)$. If there exists a $G$-symmetric state pure $\psi$ that satisfies the split property, then $\omega$ vanishes.
\end{corollary}
When the symmetry is unitary, this corollary reduces to Remark 4.1 of Ref.~\cite{kapustin2024anomalous}, which is one of the main results of Ref.~\cite{kapustin2024anomalous}. Here we extend this result to anti-unitary symmetries. Comparing Corollary \ref{corollary: cluster + area law} and Corollary \ref{corollary: Kapustin-Sopenko}, the former is often easier to use, because its statement only involves conditions like the clustering property and entanglement area law, which are relatively easy to check, thanks to the vast results in the previous literature. On the other hand, checking whether a state is pure often requires more effort, so Corollary \ref{corollary: Kapustin-Sopenko} is relatively more difficult to apply.

\begin{proof}[Proof of Corollary \ref{corollary: Kapustin-Sopenko}] Because $\psi$ is pure, according to Lemma \ref{lemma:typeI}, it is a type-I factor state. Then by the root lemma (Lemma \ref{lemma: root lemma}), $\omega$ vanishes.
    
\end{proof}

\subsection{The more standard LSM theorems:\\ Constraints on the spectrum and symmetry-enforced long-range entanglement} \label{subsec: LRE and spectrum}

Lemma \ref{lemma: root lemma}, Corollary \ref{corollary: cluster + area law} and Corollary \ref{corollary: Kapustin-Sopenko} can all be viewed as some generalized versions of the LSM theorem. Using these theorems, we can deduce the more standard LSM theorems very easily.

First, we show that an anomalous symmetry imposes nontrivial constraints on the energy spectrum of the Hamiltonian.

\begin{theorem} \label{thm: spectrum}
    Let $\tilde{\alpha}:G\to\tG^{lp}$ be a symmetry action with a nontrivial anomaly index $\omega\in\rH^{3}_{\varphi}(G;\U)$ and $H$ be an admissible $G$-symmetric Hamiltonian. Then $H$ does not admit a $G$-symmetric locally unique gapped ground state.
\end{theorem}

This theorem extends Theorem V.1 of Ref. \cite{Liu2024LRLSM} to systems with anti-unitary symmetries. There are multiple ways to prove this theorem. For example, we can prove it by using Corollary \ref{corollary: cluster + area law}. First, we note that the locally unique gapped ground states of an admissible Hamiltonian satisfy the clustering property (see Theorem 2.8 of Ref. \cite{Hastings2006decay}). Second, it is known that the locally unique gapped ground states of an admissible Hamiltonian satisfy the entanglement area law (see Theorem 4.4 of Ref. \cite{Ukai2024} and Theorem IV.1 of Ref. \cite{Liu2024LRLSM}). Therefore, according to Corollary \ref{corollary: cluster + area law}, an admissible Hamiltonian with an anomalous symmetry cannot have a $G$-symmetric locally unique gapped ground states.

Another version of the LSM theorem, which shows that an anomalous symmetry forces a system to be long-range entangled, can also be derived.

\begin{theorem} \label{thm: LRE}
   Let $\tilde\alpha: G\to\tG^{lp}$ be an anomalous $G$ symmetry action, \ie its anomaly index is nontrivial. Let $\psi$ be a $G$-symmetric state, \ie $\psi=\psi_{\tilde\alpha(g)}$ for all $g\in G$. Then $\psi$ is long-range entangled. 
\end{theorem}

In this paper, we define long-range entangled (LRE) states and short-range entangled (SRE) states following Ref. \cite{Liu2024LRLSM}. In an infinite-size system, a state is considered as an SRE state if it can be disentangled (\ie converted to a product state) by a finite-time evolution generated by an almost local Hamiltonian, otherwise it is an LRE state. Here an almost local Hamiltonian is a Hamiltonian that can contain non-local interactions, but as the spatial sizes of these interactions increase, their magnitudes decay faster than any polynomial function.

Theorem \ref{thm: LRE} provides powerful constraints on the entanglement structure of a state invariant under an anomalous symmetry, and it implies that systems with anomalous symmetries are natural platforms to look for LRE states. As discussed in Ref. \cite{Liu2024LRLSM}, because we define SRE and LRE states based on time evolutions generated by almost local Hamiltonians, rather than finite-depth quantum circuits or time evolutions generated by strictly local Hamiltonians, the symmetry-enforced long-range entanglement proved here is a very strong statement. Moreover, our results actually also imply that an anomalous symmetry is not comptible with an invertible state, which, when being stacked with another state, becomes an SRE state.

Theorem \ref{thm: LRE} can also be proved by using any of Lemma \ref{lemma: root lemma}, Corollary \ref{corollary: cluster + area law} and Corollary \ref{corollary: Kapustin-Sopenko}. For example, to prove Theorem \ref{thm: LRE} based on Lemma \ref{lemma: root lemma}, we first note that an SRE is a pure state (see, for example, footnote 43 of Ref. \cite{Liu2024LRLSM}), so it is a type-I factor state. Second, as shown in Lemma 4.2 in Ref. \cite{Kapustin2020invertible}, an SRE state satisfies the area law of entanglement. So according to Lemma \ref{lemma:split}, an SRE state satisfies the split property. Then Lemma \ref{lemma: root lemma} shows that an SRE state is not compatible with an anomalous symmetry.

\section{Consequences for finite-size systems} \label{sec: finite systems}

Our analysis above applies to infinite-size quantum spin chains. One may wonder how the main theorems apply to realistic physical systems, which, though large, are always finite. For the cases with only unitary symmetries, this question is systematically addressed in the last appendix of Ref.~\cite{Liu2024LRLSM}. Below we will describe the finite-size versions of Theorem \ref{thm: spectrum} and Theorem \ref{thm: LRE}, which generalizes the results in Ref. \cite{Liu2024LRLSM} to cases with anti-unitary symmetries.

\subsection{Setup for thermodynamic limits}

To begin, we will give a sharp definition of the thermodynamic limits of a squence of finite-size states, automorphisms and Hamiltonians.

Let $\Lambda\simeq\z$ be the infinite lattice and fix an increasing sequence of finite subsets,
\beq\label{eq:increasing}
\Gamma_{1}\subseteq\Gamma_{2}\dots \subseteq\Lambda,\quad |\Gamma_{L}|<\infty
\eeq
indexed by an integer $L$, which can be thought of as the size of a finite system. We assume each $\Gamma_{L}$ is connected and it eventually exhausts $\Lambda$, \ie
\beq\label{eq:exhausting}
\bigcup_{L=1}^{\infty}\Gamma_{L}=\Lambda
\eeq

Let $\psi_{L}:\A_{\Gamma_{L}}\to\bbC$ be a sequence of quantum states. we formally extend them to be a state of $\A^{ql}$ as follows. Let $\Omega$ be a fixed pure product state of $\A^{ql}$ (which models the environment), then $\Omega_{L}:=\Omega|_{\Gamma_{L}^{c}}$ is another pure product state on $\A_{\Gamma_{L}^{c}}$ ,where $\Omega|_{\Gamma_{L^c}}$ is the restriction of $\Omega$ on $\Gamma_{L^c}$. We define $\tilde{\psi}_{L}:=\psi_{L}\otimes\Omega_{L}$, which is a state on $\A^{ql}$. Formalizing this idea, we define a sequence of finite-size states as follows.

\begin{definition}\label{eq:sequence_states}
     Fix a pure product state $\Omega:\A^{ql}\to \bbC$, a sequence of finite-size states $\{\tilde{\psi}_{L}\}_{L=1,2,\dots}$ means
    \beq
    \tilde{\psi}_{L}:=\psi_{L}\otimes\Omega|_{\Gamma_{L}^{c}}
    \eeq
    where $\psi_{L}$ is a state of $\A_{\Gamma_{L}}$.
\end{definition}
Having a sequence of states, a natural task is to define the notion of convergence. In the following, we will only use the so-called {\it weak-$*$ convergence} of sequence of states, which physically means that the expectation values of local observables converge in the thermodynamic limit.

\begin{definition}[weak-$*$ convergence of states]\label{def:weak-$*$}
    Let $\{\tilde\psi_{L}\}$ be a sequence of states of $\A^{ql}$ defined as above. We say that this sequence of states weak-$*$ converges to a state $\psi$ in the infinite-size system if for any $\epsilon>0$ and $A\in \A^{l}$, there exists a constant $L_{\epsilon,A}$ such that $\forall\,L>L_{\epsilon,A}$, we have
    \beq
    |\tilde\psi_{L}(A)-\psi(A)|<\epsilon ||A||
    \eeq
    In this case, we say that $\psi$ is the weak-$*$ limit of $\psi_{L}$ as $L\to\infty$ and write $\lim_{L\to\infty}\psi_{L}=\psi$.
\end{definition}
In other words, weak-$*$ convergence means for any fixed finite region $\Gamma$, the reduced density matrix of $\psi$ on $\Gamma$ converges in the trace norm. This equivalent characterization does not depend on the choice of $\Omega$ (i.e. the environment).

Given a sequence of states, its weak-$*$ limit may not exist. However, as discussed in Ref.~\cite{Liu2024LRLSM}, the celebrated Banach-Alaoglu-Bourbaki theorem ensures that one can always find a weak-$*$ convergent subsequence.

It will also be useful to define the notion of convergence for Hamiltonians and automorphisms. Let $\talpha_{L}$ be a twisted $*$-automorphism of $\A_{\Gamma_{L}}$, and we formally extend it to $\A^{ql}$ by viewing it as $\talpha_{L}\otimes\mathrm{id}_{\Gamma_{L}^{c}}$ if it is linear, and as $\talpha_{L}\otimes K_{\Gamma_{L}^{c}}$ if it is anti-linear, where $K_{\Gamma_{L^c}}=\prod_{j\in\Gamma_{L^c}}K_j$. Again, by abusing notations, we will write $\talpha_{L}$ rather than $\talpha_{L}\otimes \mathrm{id}_{\Gamma_{L}^{c}}$ or $\talpha_{L}\otimes K_{\Gamma_{L}^{c}}$ . Therefore, we have a natural notion of finite-size sequence of twisted automorphisms (and similarly, sequence of finite-size Hamiltonians).
\begin{definition}\label{def:sequence_aut}
    A sequence of finite-size twisted $*$-automorphisms on $\A^{ql}$ $\{\talpha_L\}_{L=1,2,\cdots}$ means $\{\tilde\alpha_{L}\otimes\theta(\talpha_L)\}_{L=1,2,\dots}$, where each $\tilde\alpha_{L}$ is a twisted $*$-automorphism on $\A_{\Gamma_{L}}$, and $\theta(\talpha_L)=\mathrm{id}_{\Gamma_L^c}$ if $\talpha_L$ is linear, while $\theta(\talpha_L)=K_{\Gamma_L^c}=\prod_{j\in\Gamma_L^c}K_j$ if $\talpha_L$ is anti-linear. Similarly, a sequence of finite-size Hamiltonians $\{H_{L}\}_{L=1,2,\dots}$ means $\{H_{L}\otimes \mathrm{id}_{\Gamma^{c}_{L}}\}_{L=1,2,\dots}$, where each $H_{L}$ is a self-adjoint operator on $\Gamma_{L}$.
\end{definition}
However, merely a sequence does not buy us much. We now define the strong convergence of such finite-size sequences of twisted $*$-automorphisms and finite-size Hamiltonians.
\begin{definition}\label{def:strong_convergence}
    Let $\tilde{\alpha}_{L}\in\tA(\A_{\Gamma_{L}})$, and consider the sequence $\{\tilde{\alpha}_{L}\}_{L\in\mathbb{N}}$.
    If for any fixed $A\in\A^{l}$ and $\epsilon>0$, there exists $\tilde\alpha\in\tA(\A^{ql})$ and $L_{\epsilon,A}>0$, such that for any $L>L_{\epsilon,A}$,
    \beq
    ||\tilde{\alpha}_{L}(A)-\tilde{\alpha}(A)||<\epsilon ||A||
    \eeq
    Then we say $\{\tilde{\alpha}_{L}\}$ strongly converges to $\tilde{\alpha}$ as $L\to\infty$ and write $\lim_{L\to\infty}\tilde{\alpha}_{L}=\tilde{\alpha}$.

    Similarly, if for any fixed $A\in\A^{\ell}$ and $\epsilon>0$, there exists a densely defined derivation $\delta_{H}$ and $L_{\epsilon,A}>0$, such that for any $L>L_{\epsilon,A}$,
    \beq
    ||\delta_{H}(A)-[H_{L},A]||<\epsilon||A||
    \eeq
    Then we say $\ad_{H_{L}}:=[H_{L},\cdot]$ strongly converges to $\delta_{H}$ and write $\lim_{L\to\infty}\ad_{H_{L}}=\delta_{H}$. We say $H_{L}$ is admissible if $H$ is admissible.
\end{definition}

The following lemma shows that these notions of convergences fit nicely.
\begin{lemma}[Lemma G.1 of Ref.~\cite{Liu2024LRLSM}]\label{lemma:composition}
    Consider sequences of finite-size states $\{\psi_{L}\}_{L=1,2,\dots}$, finite-size $*$-automorphisms $\{\talpha_{L}\}_{L=1,2,\dots}$ and finite-size Hamiltonians $\{H_{L}\}_{L=1,2,\dots}$. Assume that 
    \beq
    \begin{split}
        \lim_{L\to\infty}\psi_{L}&=\psi\\
        \lim_{L\to\infty}\talpha_{L}&=\talpha\\
        \lim_{L\to\infty}[H_{L},\cdot]&=\delta_{H}
    \end{split}
    \eeq
    Then
    \begin{enumerate}
        \item $\lim_{L\to\infty}\psi_{L}(A)=\psi(A),\quad\forall\,A\in\A^{ql}$.
        \item $\lim_{L\to\infty}\psi_{L\talpha_{L}}=\psi_{\talpha}$.
        \item $\lim_{L\to\infty}\psi_{L}(B\,\ad_{H_{L}}(A))=\psi(B\,\delta_{H}(A)),\quad\forall\,A\in D(\delta_{H},)\,B\in\A^{ql}$.
    \end{enumerate}
    where $D(\delta_{H})\supseteq \A^{\ell}$ is the domain of $\delta_{H}$.
\end{lemma}
This lemma generalizes Lemma G.1 of Ref. \cite{Liu2024} to include anti-linear automorphisms. The proof of this lemma is straightforward and very similar to the proof of Lemma G.1 in \cite{Liu2024LRLSM}, so we do not reproduce it here.

\subsection{Consequence on the spectrum}
Now we arrive at the core of this section, \ie the finite-size versions of Theorem \ref{thm: spectrum} and Theorem \ref{thm: LRE}.

\begin{theorem}\label{thm:finite_size_LSM}
    Let $\alpha_{L}:G\to\tA(\A_{\Gamma_{L}})$ be a strongly convergent sequence of symmetry actions and $H_{L}$ be a strongly convergent sequence of Hamiltonians. If 
    \begin{enumerate}
        \item $\talpha:=\lim_{L\to\infty}\talpha_{L}$ is a symmetry action described by (twisted) LPA.
        \item $\delta_{H}:=\lim_{L\to\infty}\ad_{H_{L}}$ is an admissible Hamiltonian.
        \item $H_{L}$ is invariant under $\talpha_{L}(g)$ for all $L$ and $g\in G$.
        \item $H_{L}$ has a unique gapped ground state with uniformly lower bounded energy gap $\Delta_{L}\geqslant \Delta>0$.
    \end{enumerate}
    Then $\talpha$ has trivial anomaly index.
\end{theorem}
This theorem, where the symmetry can contain both unitary and anti-unitary symmetries, is a generalization of Theorem G.3 in Ref.~\cite{Liu2024LRLSM}, which only discusses unitary symmetries. The same proof as in Ref. \cite{Liu2024LRLSM} applies here, once the LPA discussed in Ref. \cite{Liu2024LRLSM} are properly replaced by their twisted versions. For this reason, we do not reproduce the full details here.

As a corollary, we have the finite-size version of Theorem \ref{thm: spectrum}.
\begin{corollary}\label{corollary:finite_LSM}
    If a sequence of finite-size spin chains have admissible Hamiltonians $H_{L}$ which are symmetric under an anomalous symmetry, then $H_{L}$ cannot have a unique gapped ground state for sufficiently large $L$.
\end{corollary}
Note that the maximal size of such systems that allows a unique gapped ground state depends on the details of $H_{L}$ and $\alpha_{L}$, and is hence non-universal.

\begin{remark}
    Although we focus on open boundary conditions in above discussion, our Theorem \ref{thm:finite_size_LSM} and Corollary \ref{corollary:finite_LSM} hold valid for systems with periodic boundary conditions as well (see Appendix G.5 of Ref.~\cite{Liu2024LRLSM} for detailed discussion on the periodic boundary condition).
\end{remark}

\subsection{Symmetry-enforced long-range entanglement in finite-size systems}

Next, we turn to symmetry-enforced long-range entanglement, \ie any state invariant under an anomalous symmetry must be long-range entangled, as stated in Theorem \ref{thm: LRE}.

We will present the finite-size version of symmetry-enforced long-range entanglement, which is much more nontrivial and stronger than Theorem \ref{thm: LRE}. In this context, we say that a sequence of finite but large systems is SRE if each of them can be disentangled by a time evolution generated by an almost local Hamiltonian over a duration that does not diverge as the system size goes to infinity, otherwise it is LRE. Then we have the following finite-size version of Theorem \ref{thm: LRE}.
\begin{theorem}\label{thm:finite_SRE_main}
    Let $\tilde\alpha_L$ be an operation on a system with size $L$, and suppose that $\tilde\alpha_L$ converges to an anomalous symmetry $\tilde\alpha: G\to\tG^{lp}$ as $L\rightarrow\infty$. Let $\{|\psi_L\ra\}$ be a sequence of SRE states defined on a system of size $L$. Then $|\psi_{L}\ra$ is symmetric under $\alpha_L$ for at most finitely many $L$.
\end{theorem}

This theorem extends Theorem VI. 2 in Ref. \cite{Liu2024LRLSM}, which only considers unitary anomalous symmetries, to situations that can include anti-unitary anomalous symmetries. Its proof is also very similar to that of  Theorem VI. 2 in Ref. \cite{Liu2024LRLSM}. To prove this theorem, we first note that such a sequence of states must have a subsequence that has a well-defined thermodynamic limit, due to the Banach-Alaoglu-Bourbaki theorem. This thermodynamic limit can be shown to satisfy the clustering property and entanglement area law (see the proof of Theorem G.5 in Ref. \cite{Liu2024LRLSM}). Hence a symmetric SRE can be ruled out by the powerful Corollary \ref{corollary: cluster + area law}.

\begin{remark}
    In the above theorem, we do not need to assume that this sequence of finite-size states converge to an infinite-size version of SRE state, and this is why this theorem is so nontrivial and strong. In fact, it is currently not known if our definition on SRE sequences ensures that its weak-$*$ limit (if exists) is an infinite-size version of SRE state, since the convergence of states does not imply the convergence of disentanglers even if we restrict ourselves to a subsequence.
\end{remark}

\section{Applications} \label{sec: applications}

After proving our main theorems, we present some representative applications of them, where anti-unitary symmetries play an important role.

\subsection{Quantum spin chains with Dyzaloshinski-Moriya interactions}

We first consider a quantum spin chain with a Dyzaloshinski-Moriya (DM) interaction, whose Hamiltonian is
\beq
H=\sum_{i, j}\left(J_{ij}\vec S_i\cdot\vec S_j+\vec D_{ij}\cdot(\vec S_i\times\vec S_j)\right)
\eeq
where $\vec S_i$ is the spin operator at the site $i$, and $J_{ij}$ and $\vec D_{ij}$ are coupling constants that are assumed to decay faster than $\left({\rm dist}(i, j)\right)^{-2}$. In this Hamiltonian, the first Heisenberg term preserves the $SO(3)$ spin rotational symmetry and the $\z_2^T$ time reversal symmetry that takes $\vec S_i\rightarrow-\vec S_i$. The second term is the DM interaction arising from spin-orbit coupling, which breaks the $SO(3)$ symmetry but preserves the $\z_2^T$ symmetry. It is further assumed that this Hamiltonian is translation invariant. So this Hamiltonian has a $G=\z_2^T\times\z$ symmetry.

As demonstrated in Sec. \ref{subsec: example}, this $G$ symmetry is anomaly free if the local moments are Kramers singlets, and it is anomalous if the local moments are Kramers doublets. In the latter case, we can conclude that the Hamiltonian cannot have a unique gapped symmetric ground state, according to Theorem \ref{thm: spectrum} and Corollary \ref{corollary:finite_LSM}, and any state invariant under the $G$ symmetry must be long-range entangled, according to Theorem \ref{thm: LRE} and Theorem \ref{thm:finite_SRE_main}. Moreover, by Corollary \ref{corollary: cluster + area law}, a $G$-symmetric state in such a system cannot satisfy both the clustering property and entanglement area law.

\subsection{Chiral spin chains with an anti-unitary translation symmetry} \label{subsec: anti-unitary translation}

In Ref. \cite{Yao2023}, a class of models for spin-$S$ chiral spin chains are studied, with the Hamiltonian
\beq \label{eq: chiral spin chains}
H=\sum_{i\alpha}J_\alpha S^\alpha_iS^\alpha_{i+1}+\sum_{i\alpha\beta\gamma}(-1)^iK_{\alpha\beta\gamma}\epsilon_{\alpha\beta\gamma}S^\alpha_iS^\beta_{i+1}S^\gamma_{i+2}
\eeq
For general values of $J_\alpha$ and $K_{\alpha\beta\gamma}$, this model does not have the full $SO(3)$ spin rotational symmetry. Instead, it has a unitary $\z_2^x\times\z_2^z$ symmetry, which can be formally generated by $X=\Ad_{R^x}$ and $Z=\Ad_{R^z}$, with
\beq
R^x=\prod_ie^{i\pi S_i^x},
\quad
R^z=\prod_ie^{i\pi S_i^z}.
\eeq
Interestingly, there is also an anti-unitary translation symmetry $\z^T$, which is generated by $T_2=T_1T$, where $T_1$ is the ordinary linear translation, and $T_2$ is the the usual time reversal action for spins, under which $\vec S_i\rightarrow-\vec S_i$.
As the combination of $T_1$ and $T$, $T_2$ acts on the spin operators as
\beq
T_2(\vec S_i)=-\vec S_{i+1}.
\eeq

Based on twisted boundary condition and robustness of spectrum, Ref. \cite{Yao2023} argues that such models cannot have a unique symmetric gapped ground state if $S\in\z+\frac{1}{2}$. Below we show that the $G=\z_2^x\times\z_2^z\times\z^T$ symmetry has a nontrivial (resp. trivial) anomaly index if $S\in\z+\frac{1}{2}$ (resp. $S\in\z$). Then we will comment on the consequences of the anomaly.

The simplest way to obtain the anomaly index in the present case is to note that $G$ is a subgroup of $\z_2^x\times\z_2^z\times\z_2^T\times\z$, where $\z_2^T$ is viewed as the usual time reversal symmetry with generator $T$, and $\z$ is viewed as the ordinary linear translation symmetry with generator $T_1$. There is a natural inclusion map $i: G\to\z_2^x\times\z_2^z\times\z_2^T\times\z$. Specifically, denote the group element $(R^x)^x(R^z)^z T_2^n\in G$ by $(x, z, n)$, where $x\in\{0, 1\}$, $z\in\{0, 1\}$ and $n\in\z$. Then under the map $i$ it becomes $i[(x, z, n)]=(x, z, n, n):=(R^x)^x(R^z)^zT^tT_1^n\in \z_2^x\times\z_2^z\times\z_2^T\times\z$. Suppose that the system originally has the $\z_2^x\times\z_2^z\times\z_2^T\times\z$ symmetry, which is later broken to the $G$ symmetry. According to Eq. \eqref{eq: anomaly 3-cocycle result}, if $S\in\z$, the anomaly index of the virtual original $\z_2^x\times\z_2^z\times\z_2^T\times\z$ is already trivial, so the anomaly of the $G$ symmetry, which is its pullback induced by the map $i$, is also trivial. On the other hand, if $S\in\z+\frac{1}{2}$, the anomaly index of this virtual original $\z_2^x\times\z_2^z\times\z_2^T\times\z$ symmetry would be
\beq
\omega_0((x_1, z_1, t_1, n_1), (x_2, z_2, t_2, n_2), (x_3, z_3, t_3, n_3))=(-1)^{(x_1z_2+t_1t_2)n_3}.
\eeq
Therefore, when $S\in\z+\frac{1}{2}$, the anomaly index of the $G$ symmetry, which is the pullback of the above $\omega_0$ induced by the inclusion map $i$, is
\beq \label{eq: anomaly index for anti-unitary translation}
\begin{split}
\omega((x_1, z_1, n_1), (x_2, z_2, n_2), (x_3, z_3, n_3))&=(i^*\omega_0)((x_1, z_1, n_1), (x_2, z_2, n_2), (x_3, z_3, n_3))\\
&=\omega_0(i[(x_1, z_1, n_1)], i[(x_2, z_2, n_2)], i[(x_3, z_3, n_3)])\\
&=(-1)^{(x_1z_2+n_1n_2)n_3}
\end{split}
\eeq
One can check that this $\omega$ indeed satisfies the 3-cocycle equation, Eq. \eqref{eq: 3-cocycle equation main}. Moreover, it is shown in Appendix \ref{app: nontrivial anomaly} that Eq. \eqref{eq: anomaly index for anti-unitary translation} indeed represents a nontrivial element in $H_\varphi^3(G; \U)$. 

Therefore, when $S\in\z+\frac{1}{2}$ we can apply our theorems to extract some physical consequences. In particular, we can conclude that not only the Hamiltonian Eq. \eqref{eq: chiral spin chains} but also its long-range interacting versions satisfying the admissible condition Eq. \eqref{eq:admissible_H} cannot have a unique gapped symmetric ground state, according to Theorem \ref{thm: spectrum} and Corollary \ref{corollary:finite_LSM}, and any state invariant under the $G$ symmetry must be long-range entangled, according to Theorem \ref{thm: LRE} and Theorem \ref{thm:finite_SRE_main}. Moreover, by Corollary \ref{corollary: cluster + area law}, a $G$-symmetric state in such a system cannot satisfy both the clustering property and entanglement area law.

\section{Discussion} \label{sec: discussion}

In this work, we extend the previous discussions on the anomalies of unitary symmetries and their consequences to systems  that can have anti-unitary symmetries, by introducing twisted locality-preserving automorphisms, which describe unitary and anti-unitary symmetries in a unified fashion. We find that anomalous symmetries have important implications on the correlation, entanglement, and energy spectra of the quantum spin chains, as summarized by various generalized LSM-like theorems. These general results are demonstrated in various concrete examples.

We remark that our results apply to a very broad class of systems, whose symmetry can include both unitary and anti-unitary symmetries, on-site or non-on-site symmetries, internal and translation symmetries. The Hamiltonian of the system can be either local or non-local, as long as it satisfies Eq. \eqref{eq:admissible_H}. Moreover, we derive results for both infinite-size quantum spin chains and sequences of large but finite-size quantum spin chains.

It is important to note that our discussions of symmetries and their anomalies are entirely formulated using microscopic, Hamiltonian-independent data of a lattice system, and we do not assume any description of the system in terms of any effective theory, such as a quantum field theory. In many of the previous studies, certain assumptions about the low-energy and large-distance physics of the quantum many-body system are assumed, which often lack rigorous microscopic justification. We believe that the ultimate understanding of quantum matter should be formulated using theoretical frameworks that have the same spirit as the one in this paper, which describes the universal properties of a quantum many-body system in terms of its microscopic data.

Below we comment on some open questions for future studies.

\begin{itemize}

    \item So far our framework does not capture point-group symmetry, because such a symmetry can map a local operator to another one with a support very far away from the original one, which cannot be realized by a twisted locality-preserving automorphism. It is useful to extend our framework to also include these symmetries.

    \item Our discussion focuses on quantum spin chains. It is interesting to generalize our results to fermionic systems and bosonic systems where the dimension of each local Hilbert space is infinite.

    \item The overarching goal is to have a unified framework to describe the symmetries and anomalies for all quantum many-body systems, including systems in higher dimensions, systems living in the continuum, etc.

    \item We derive nontrivial constraints on the system due to an anomalous symmetry. However, even if a symmetry is not anomalous, there can still be nontrivial constraint if the system is restricted to a specific symmetry sector. For example, if a system has a $U(1)\times\z$ symmetry, where $U(1)$ is on-site and $\z$ is translation, then this symmetry is anomaly free. However, as the original LSM theorem states, as long as the filling factor is not an integer, there are still nontrival constraints due to this $U(1)\times\z$ symmetry \cite{Lieb1961}. A more recent example concerns a system with only translation symmetry, which is anomaly free, and it is demonstrated in Ref. \cite{Gioia2021} that a translation symmetric system with a nonzero momentum must be long-range entangled.{\footnote{The notion of long-range entanglement discussed in Ref. \cite{Gioia2021} is different from the one discussed in the present paper.}} It is interesting to incorporate these examples in a unified framework.

    \item In Ref. \cite{Zou2021}, a hypothesis of emergibility was proposed, which conjectures that the landscape of quantum phases of matter that can emerge in a lattice system is determined by the symmetry and anomaly of the lattice system. This hypothesis was then applied to carry out the classifications of various quantum phases of matter \cite{Ye2021a, Ye2023, Liu2024}. It is useful to have a microscopic justification of this hypothesis.

\end{itemize}

\begin{acknowledgements}

We thank Shiyu Zhou for a previous collaboration, and thank Shang-Qiang Ning for helpful discussion. Research at Perimeter Institute is supported in part by the Government of Canada
through the Department of Innovation, Science and Industry Canada and by the Province of Ontario through
the Ministry of Colleges and Universities. LZ is supported in part by the National University of Singapore start-up
grants A-0009991-00-00 and A-0009991-01-00. RL is also supported by the Simons Collaboration on Global Categorical Symmetries through Simons Foundation grant 888996.
    
\end{acknowledgements}

\clearpage
\appendix

\section{Anti-linear maps and complex conjugation}\label{sec:comp_conj}

Anti-linear maps and complex conjugations are extensively used in this paper. In this appendix, we provide a brief review of these concepts (see Refs.~\cite{Uhlmann16anti_linear,SWANSON2019CPT} for more information).

\subsection{Complex conjugation of vector spaces}

To deal with anti-linear maps, an important idea is to reduce anti-linear maps to linear maps. This way of thinking should be familiar from the treatment of anti-unitary symmetries in ordinary quantum mechanics, but here we present a more formal discussion that is useful for some of our purposes. 

First, we introduce the complex conjugate of a vector space $V$ over $\bbC$.
\begin{definition}[Complex conjugate vector space]
    Given a vector space $V$ over $\bbC$, its complex conjugate, $\bar{V}$, is another vector space over $\bbC$, which has
    \begin{enumerate}
        \item the same underlying set as $V$. In other words, if $v\in V$, then $v\in\bar{V}$.
        \item the same addition as $V$, \ie $V$ and $\bar{V}$ are identical as Abelian groups.
        \item a different scalar multiplication defined as $\lambda\star v:=\bar{\lambda}v$, where we have used the scalar multiplication of $V$ on the right hand side.
    \end{enumerate}
\end{definition}
From this point of view, an anti-linear operator $T$ on $V$ is nothing but a \textit{$\bbC$-linear} map $T:V\to\bar{V}$. To see it, note that for such a $\bbC$-linear map $T:V\to\bar{V}$,
\beq
T(\lambda v)=\lambda\star T(v)=\bar{\lambda}T(v),
\eeq
hence $T$ can be viewed as an anti-linear map on $V$ in the usual sense. Alternatively, an anti-linear operator can also be viewed as a linear map $\bar{T}:\bar{V}\to V$, defined by
\beq
\bar{T}(\lambda\star v):=\lambda T(v).
\eeq
An anti-linear operator $T$ on $V$ is called a complex conjugation (also known as a real structure)\footnote{What we really mean by $T^{2}=\id_{V}$ is that $T\bar{T}=\id_{\bar{V}}$ and $\bar{T}T=\id_{V}$.} if $T^{2}=\id_{V}$. In the case of a finite dimensional vector space, such operators always exist. One way to explicitly construct a complex conjugation is to choose a complete set of orthonormal basis and demand that an anti-linear map takes each of these basis states to itself. It can be verified that such an anti-linear map is a complex conjugation. However, just as there is no natural choice for the complete set of orthonomal basis, there is also no natural choice for the complex conjugation operator. Physically, this means that the ``complex conjugation operator" is always basis-dependent.
More generally, we have
\begin{corollary}\label{coro:functoriality}
    Taking complex conjugate $\bar{\cdot}$ is a functor on the category of finite-dimensional complex vector spaces $\textbf{Vec}_{\bbC}$. Furthermore, this functor is monoidal (\ie it preserves tensor product).
\end{corollary}
\begin{proof}
    To check the functoriality, we have to show that for any map $f:V\to W$ between objects of $\Vc$, there is an induced map $\bar{f}:\bar{V}\to\bar{W}$. This map can be constructed as follows,
    \beq
        \bar{f}(v)&:=f(v),\quad\forall v\in V
    \eeq
    Obviously $\bar{f}$ defined above preserves additions. For scalar multiplication, we have
    \beq
    \begin{split}
        \bar{f}(\lambda\star v)&=\bar{f}(\bar{\lambda}v)\\
        &=f(\bar{\lambda} v)\\
        &=\bar{\lambda}f(v)\\
        &=\lambda\star \bar{f}(v)
    \end{split}
    \eeq
    Hence $\bar{f}:\bar{V}\to\bar{W}$ is a linear map induced by $f:V\to W$. Moreover, it is easy to see $\bar{\id_{V}}=\id_{\bar{V}}$ and $\overline{f\circ g}=\bar{f}\circ \bar{g}$ for linear maps $f:V\to W$ and $g:U\to V$. So $\bar{\cdot}$ is indeed a functor.

    Next, we check that it is monoidal. More precisely, we show there is a natural isomorphism
    \beq
    \overline{V\otimes W}\stackrel{\sim}{\longrightarrow}\bar{V}\otimes\bar{W}.
    \eeq
    The natural isomorphism is constructed as follows. Since $\overline{V\otimes W}$ and $\bar{V}\otimes\bar{W}$ are both generated by elements of the form $v\otimes w$, we only have to define the isomorphism on these generators:
    \beq
    F(v\otimes w)=v\otimes w
    \eeq
    Obviously $F$ extends additively. The only thing to check is scalar multiplication. Note
    \beq
    \begin{split}
        F(\lambda\star (v\otimes w))&=F(\bar{\lambda}(v\otimes w))\\
        &=F((\lambda\star v)\otimes w)\\
        &=(\lambda\star v)\otimes w\\
        &=v\otimes(\lambda\star w)
    \end{split}
    \eeq
    So $F$ extends to a linear map between $\overline{V\otimes W}$ and $\bar{V}\otimes \bar{W}$. It is obvious that $F$ is injective and surjective. Hence it is an isomorphism.
\end{proof}

The above statement has an important consequence.
\begin{corollary}
    Let $T_{k},\,k=1,2$ be anti-linear operators on $V_{k},\,k=1,2$, then $T_{1}\otimes T_{2}$ is an anti-linear operator on $V_{1}\otimes V_{2}$.
\end{corollary}
\begin{proof}
    By assumption, we have $T_{k}:\bar{V}_{k}\to V_{k}$, so
    \beq
    T_{1}\otimes T_{2}:\bar{V}_{1}\otimes \bar{V}_{2}\to V_{1}\otimes V_{2}
    \eeq
    On the other hand, we know that $F:\overline{V\otimes W}\stackrel{\sim}{\longrightarrow}\bar{V}\otimes\bar{W}$. So $(T_{1}\otimes T_{2})\circ F: \overline{V_{1}\otimes V_{2}}\to V_{1}\otimes V_{2}$ is an anti-linear operator on $V_{1}\otimes V_{2}$. In practice, we often omit $F$ because it is a natural isomorphism and say $T_{1}\otimes T_{2}$ is an anti-linear operator on $V_{1}\otimes V_{2}$.
\end{proof}
A further remark is 
\begin{remark}
    Just like $V$ is naturally isomorphic to $V^{**}$, $V$ is naturally isomorphic to $\bar{\bar{V}}$. So we identify them below.
\end{remark}

One must note that there are maps which are neither $\bbC$-linear nor anti-linear. For example, let $f:V\to V$ be a $\bbC$-linear operator while $g:W\to\bar{W}$ be an anti-linear operator. Then, $f\otimes g:V\otimes W\to V\otimes \bar{W}$ is neither $\bbC$-linear nor anti-linear as an operator on $V\otimes W$. However, it does make sense as a $\bbC$-linear map from $V\otimes W$ to $V\otimes \bar{W}$. More explicitly,
\beq
\begin{split}
    (f\otimes g)( (\lambda v)\otimes w)&=f(\lambda v)\otimes g(w)\\
    &=(\lambda f(v))\otimes g(w)\\
    &=f(v)\otimes (\lambda \star g(w))\\
    &=f(v)\otimes g(\lambda w)\\
    &=(f\otimes g)(v\otimes \lambda w)
\end{split}
\eeq
which is consistent.

\subsection{Complex conjugation on $C^*$-algebra}

Similarly, for a $C^*$-algebra $\A$, we can define its complex conjugate algebra $\bar{\A}$.
\begin{definition}[Complex conjugate $C^*$-algebra] \label{def: conjugate algebra}
    Given a $C^{*}$-algebra $\A$, for its complex conjugate algebra, $\bar{ \A}$,
    \begin{enumerate}
        \item the underlying vector space is the complex conjugate of $\A$.
        \item the algebra multiplication is given by
        \beq
        \bar{\A}\otimes\bar{\A}\stackrel{\sim}{\to}\overline{\A\otimes\A}\ \stackrel{\bar{\mu}}{\to} \bar{\A}.
        \eeq
        where the last step is induced by multiplication $\mu:\A\otimes\A\to\A$.
        \item the norm and $*$ (\ie $\dagger$) structure is the same as $\A$.
    \end{enumerate}
\end{definition}
The benefit to introduce $\bar{\A}$ is that anti-linear automorphisms of $\A$ are nothing but $*$-isomorphisms between $\A$ and $\bar{\A}$, just like an anti-linear operator on $V$ is a linear map between $V$ and $\bar{V}$. Moreover, let $\alpha$ be an anti-linear automorphism of $\A$ and $\psi$ be a state of $\A$, then $\psi\circ\alpha$ is \textit{not} a state of $\A$ but a state of $\bar{\A}$. Similarly, $\bar{\psi}$, which is defined by $\bar{\psi}(A):=\overline{\psi(A)}$, is not a state of $\A$ but a state of $\bar{\A}$.

In general, one cannot expect $\A\simeq \bar{\A}$ as complex $C^{*}$-algebra {\footnote{But they are isomorphic over $\R$.}} (see Ref.~\cite{Rosenberg2015realC*alg} and references therein for more information). Equivalently, for a general complex $C^*$-algebra, there may not be an anti-linear operator $T$ such that $T^{2}=1$ (\ie a complex conjugation operator). In the mathematical literature, if $\A\simeq \bar{\A}$ as complex $C^{*}$-algebra (\ie it admits a complex conjugation operator), then $\A$ is said to be symmetric{\footnote{However, this name can be confusing and misleading in our context, because here we are not discussing a symmetry action at all!}}.
\begin{definition}
    A complex $C^{*}$-algebra $\A$ is said to be symmetric if 
    \beq
    \A\simeq\bar{\A}
    \eeq
    over $\bbC$. Later on, by a symmetric $C^*$-algebra, we mean the algebra $\A$ together with an isomorphism $K:\A\to\bar{\A}$.
\end{definition}
It is natural to ask if $\A^{ql}$ is symmetric, and the answer is yes.
\begin{theorem}\label{theorem:symmetric}
    The algebra of quasi-local operators $\A^{ql}$ is symmetric\footnote{More generally, one can show that uniformly hyperfinite (UHF) algebras (which includes $\A^{ql}$) and approximately finite-dimensional (AF) algebras (which includes UHF algebras) are all symmetric.}.
\end{theorem}
To prove above theorem, we need following two lemmas.
\begin{lemma}\label{lemma:comp_conj-fd}
    A finite-dimensional vector space $V$ over $\bbC$ always admits a complex conjugation operator.
\end{lemma}
\begin{proof}
      Note that there is always a linear isomorphism $F:V\to\bbC^{N}$ where $N=\dim(V)$. Let $\sigma$ be complex conjugation on $\bbC^{N}$, then $\sigma_{F}:=F^{-1}\circ\sigma\circ F$ gives our desired complex conjugation.
\end{proof}
\begin{lemma}[Bounded linear transformation theorem, Theorem A.36 of Ref.~\cite{hall2013quantum}]\label{lemma:BLT}
    Let $F:V_{1}\to V_{2}$ be a bounded linear map between normed vector spaces $V_1$ and $V_2$, where $V_{2}$ is complete. Then there is a unique extension $\hat{F}$ of $F$ defined on $\hat{V}_{1}$ (the completion of $V_{1}$ with respect to its norm).
\end{lemma}

Below we give a proof of Theorem \ref{theorem:symmetric}.

\begin{proof}[Proof of Theorem \ref{theorem:symmetric}]
    We first define a complex conjugation on $\A^{l}$, and then extend it to $\A^{ql}$. 
    
    For each site $n\in\z$, the algebra $\A_{n}$ is isomorphic to $M_{m_{n}}(\bbC)$ ($m_{n}\times k_{n}$ matrix algebra over $\bbC$) for some $m_{n}\in\z^{\geqslant 0}$. By Lemma \ref{lemma:comp_conj-fd}, there is a complex conjugation operator $K_{n}$ defined on $M_{m_{n}}(\bbC)$. A complex conjugation operator on $\A^{l}$ is obtained by the formal tensor product
    \beq\label{eq:colimit}
    K:=\otimes_{n\in\z}K_{n}
    \eeq
    This infinite tensor product is unambiguous when it acts on $\A^{ql}$.

    Furthermore, $K:\A^{l}\to\bar{\A}^{l}$ can be viewed as another map (still denoted by $K$) $K:\A^{l}\to\bar{\A}^{l}\hookrightarrow\bar{\A}^{ql}$. Note that $K$ is bounded since $K^{2}=1$ and $\bar{\A}^{ql}$ is complete by construction (\ie it is defined by taking the norm completion of $\bar{\A}^{l}$). By Lemma \ref{lemma:BLT}, there is a unique extension $\hat{K}$ of $K$ which gives a map $\hat{K}:\A^{ql}\to\bar{\A}^{ql}$. Notice $K^{2}=1$ implies $\hat{K}^{2}=1$ by the uniqueness. Hence $\hat{K}$ gives a complex conjugation map on $\A^{ql}$.

    This proof generalizes directly to the case of uniformly hyperfinite (UHF) algebras and approximately finite-dimensional (AF) algebras.
\end{proof}

From now on, we shall not distinguish $\hat{K}$ and $K$. Note that $K:\A^{ql}\stackrel{\sim}{\longrightarrow}\bar{\A}^{ql}$ is an isomorphism. By fixing the choice of $K$, one can identify $\A^{ql}$ and $\bar{\A}^{ql}$. Moreover, $K$ preserves locality in the most strict sense because it is on-site.

\subsection{Complex conjugation of GNS representations}

For a $C^*$-algebra $\A$ (\eg $\A^{ql}$), given its state $\psi$ one can construct the GNS representation $\pi_{\psi}$ associated with $\psi$ (see \eg Refs.~\cite{Naaijkens_2017, Landsman:2017hpa,Liu2024LRLSM} for more details). As recalled in the main text, the GNS triple $(\cH_{\psi},\pi_{\psi},|\psi\ra)$ is given by 
\begin{enumerate}
    \item A Hilbert space $\cH_{\psi}$.
    \item A $*$-homomorphism $\pi_{\psi}:\A\to\B(\cH_{\psi})$.
    \item A vector $|\psi\ra$ in $\cH_{\psi}$ such that
    \beq
    \la\psi|\pi_{\psi}(A)|\psi\ra=\psi(A),\quad \forall\,A\in\A.
    \eeq
\end{enumerate}

Given the state $\psi$, the GNS triple is unique up to unitary equivalence. Before diving into the complex conjugate representation, we remark here on the complex conjugation of Hilbert spaces.
\begin{lemma}[Riesz representation theorem]
    Let $\cH$ be a Hilbert space, then there is a natural isomorphism $\cH\simeq \bar{\cH}$.
\end{lemma}
For a proof, see, \eg corollary 6.4.2 of Ref.~\cite{oden2018applied}.

Now, note that given a GNS triple $(\cH_{\psi},\pi_{\psi},|\psi\ra)$ of $\A$, one can take its complex conjugate to obtain a GNS triple $(\bar{\cH}_{\psi},\bar{\pi}_{\psi},|\psi\ra)$ associated with the algebra $\bar{\A}$ and the state $\bar{\psi}$. Furthermore, we can use Riesz representation theorem to identify $\bar{\cH}_{\psi}$ with $\cH_{\psi}$.
\begin{definition}
    The representation $(\bar{\cH}_{\psi},\bar{\pi}_{\psi},|\psi\ra)$ of $\bar{\A}$ is called the complex conjugate representation of $(\cH_{\psi},\pi_{\psi},|\psi\ra)$ (of $\A$).
\end{definition}
\begin{remark}
    It is easily verified that $(\bar{\pi}_{\psi},\bar{\cH}_{\psi},|\psi\ra)$ is the GNS representation of $\bar{\psi}$ as a state of $\bar{\A}^{ql}$.
\end{remark}

Moreover, if $\A$ is symmetric (\ie $\A\simeq \bar{\A}$), then one can identify $\A$ and $\bar{\A}$ by choosing a complex conjugation operator $K:\bar{\A}\to \A$. Thus, $\bar{\pi}_{\psi}\circ K$ gives a representation of $\A$. Note that, a priori, there is no reason to have $\pi_{\psi}\simeq \bar{\pi}_{\psi}\circ K$.

\subsection{Anti-unitary equivalence}\label{app:Jordan}

For our purpose, the notion of anti-unitary equivalence is crucial (see Definition~\ref{def:anti-unitary} for definitions). The following lemma is of fundamental importance in anti-unitary equivalence.
\begin{lemma}\label{lemma:canonical_anti_uni}
     Let $\pi_{\psi}$ and $\bar{\pi}_{\psi}$ be a pair of complex conjugate representations of $
     \A$ and $\bar{\A}$, respectively. Then there is a canonical anti-unitary map $J:\cH_\psi\to\bar{\cH}_{\psi}$ such that
     \beq
    \bar{\pi}_{\psi}(A)=J\pi_{\psi}(A)J^{-1}
    \eeq
\end{lemma}
\begin{proof}
    To see this, note that $\id:\cH_{\psi}\to\cH_{\psi}$ is equivalent to an anti-linear map $J:\cH_{\psi}\to\bar{\cH}_{\psi}$. That is, for $v\in\cH$, we have\footnote{Note that there is no natural \textit{linear} map from $V\to\bar{V}$ in general, but there is a natural \textit{anti-linear} map $J:V\to\bar{V}$ induced by identity on $V$.}
\beq
J(\lambda v)=\lambda v=\bar{\lambda}*J(v)
\eeq
It is easy to check that
\beq
\bar{\pi}_{\psi}(A)=J\pi_{\psi}(A)J^{-1}
\eeq
\end{proof}

Now we come to the proof of Lemma \ref{lemma:Jordan_equiv} in the main text (restated below).
\begin{lemma}\label{lemma:jordan2}
    Let $\psi$ be a factor state of $\A^{ql}$ and $\tilde{\alpha}\in\tA(\A^{ql})$. Then the following two conditions are equivalent.
    \begin{enumerate}
        \item There is a quasi-equivalence between $\psi$ and $\psi_{\tilde{\alpha}}$, where $\psi_{\tilde{\alpha}}$ is defined by
        \beq
        \psi_{\tilde{\alpha}}(A):=\begin{cases}
            \psi(\tilde{\alpha}(A)),\quad\text{for linear $\tilde{\alpha}$};\\
            \psi(\tilde{\alpha}(A)^*),\quad\text{for anti-linear $\tilde{\alpha}$.}
        \end{cases}
        \eeq
        \item     For any $\epsilon>0$, there is a finite subset $\Gamma_{\epsilon}$ of the lattice, such that  
        \beq
        |\psi(A)-\psi_{\tilde{\alpha}}(A)|<\epsilon||A||,\quad\forall\,A\in\A^{ql}_{\Gamma^{c}_{\epsilon}},
        \eeq
        where $\Gamma_{\epsilon}^{c}$ means the complement of $\Gamma_{\epsilon}$.
    \end{enumerate}
\end{lemma}
To prove this, we need another lemma.
\begin{lemma}[Proposition 3.2.8 of Ref.~\cite{Naaijkens_2017}]\label{lemma:equiv}
    Let $\omega_{1}$ and $\omega_{2}$ be two factor states of $\A^{ql}$, then the following two statements are equivalent.
    \begin{enumerate}
        \item States $\omega_{1}$ and $\omega_{2}$ are quasi-equivalent.
        \item For any $\epsilon>0$, there exists a finite region $\Gamma_{\epsilon}$ of the lattice, such that
        \beq
        |\omega_{1}(X)-\omega_{2}(X)|<\epsilon ||X||,\quad \forall\,X\in\A^{ql}_{\Gamma_{\epsilon}^{c}}.
        \eeq
    \end{enumerate}
\end{lemma}
An important observation is that this lemma is also true for $\bar{\A}^{ql}$ since it is isomorphic to $\A^{ql}$ (unnaturally, though), so it is also quasi-local.
\begin{proof}[Proof of Lemma \ref{lemma:jordan2}]
    If $\tilde{\alpha}$ is linear, then Lemma \ref{lemma:jordan2} reduces to Lemma \ref{lemma:equiv}. In the case of an anti-linear $\tilde{\alpha}$, first suppose we have
    \beq
    |\psi(A)-\psi_{\tilde{\alpha}}(A)|<\epsilon ||A||,\quad\forall\,A\in\A^{ql}_{\Gamma^{c}_{\epsilon}}.
    \eeq
    This implies
    \beq
    |\bar{\psi}(A)-\psi(\tilde{\alpha}(A))|<\epsilon||A||,\quad\forall\,A\in\A^{ql}_{\Gamma^{c}_{\epsilon}}.
    \eeq
    Now we use Lemma \ref{lemma:equiv} for $\bar{\A}^{ql}$ with states $\bar{\psi}$ and $\psi\circ\tilde{\alpha}$, and we see that there is a unitary equivalence
    \beq
    U\bar{\pi}_{\psi}(A) U^{-1}=\pi_{\psi}(\tilde{\alpha}(A))
    \eeq
    By Lemma \ref{lemma:canonical_anti_uni}, there is an anti-unitary map $J$ such that
    \beq
    \bar{\pi}_{\psi}(A)=J\pi_{\psi}(A)J^{-1}
    \eeq
    Thus by defining $\tilde{U}=U J$, we obtain an anti-unitary equivalence $\pi_{\psi}\approx\pi_{\psi}\circ\tilde{\alpha}$.

    Conversely, if $\pi_{\psi}$ is anti-unitarily equivalent to $\pi_{\psi}\circ\tilde{\alpha}$, then by applying Lemma \ref{lemma:equiv} to $\bar{\A}^{ql}$ with states $\bar{\psi}$ and $\psi\circ\tilde{\alpha}$, we obtain that 
    \beq
        |\bar{\psi}(A)-\psi(\tilde{\alpha}(A))|<\epsilon||A||,\quad\forall\,A\in\A^{ql}_{\Gamma^{c}_{\epsilon}},
    \eeq
    which implies 
    \beq
     |\psi(A)-\psi_{\tilde{\alpha}}(A)|<\epsilon ||A||,\quad\forall\,A\in\A^{ql}_{\Gamma^{c}_{\epsilon}}.
    \eeq
    This completes the proof.
\end{proof}

Moreover, we can deduce Corollary \ref{corollary:anti-unitary_map} from the above proof, where the $\tilde{U}$ in Corollary \ref{corollary:anti-unitary_map} is given above as $\tilde{U}=UJ$.

\section{Index theory of twisted locality-preserving automorphisms}\label{app:index}

In this appendix, we present some details about the GNVW index of $\tG^{lp}$.

\subsection{Brief review of GNVW index of quantum cellular automata}\label{sec:QCA_review}

Before discussing the GNVW index of twisted LPA, we first outline the construction of GNVW index of QCA's in this section. Our treatment follows Ref.~\cite{Ranard_2022} and more details can be found there.

Consider a spin chain which has $\A_{n}$ as its operator algebra at site $n\in\z$. Let $\alpha\in\G^{QCA}$ be a QCA. By re-grouping the lattice sites, we can always assume $\alpha$ to be a nearest-neighbor QCA, \ie
\beq
\alpha(\A_{n})\subset \A_{n-1}\otimes\A_{n}\otimes\A_{n+1}
\eeq
It can be shown that the GNVW index does not depend on the choice of re-grouping \cite{Gross_2012}.

Consider
\beq\label{eq:block_alg}
\begin{split}
    \B_{n}&:=\A_{2n}\otimes\A_{2n+1},\\
    \cC_{n}&:=\A_{2n-1}\otimes\A_{2n}.
\end{split}
\eeq
By assumption, $\alpha$ is nearest-neighbor. We then have $    \alpha(\B_{n})\subset\cC_{n}\otimes\cC_{n+1}$. If we define
\beq\label{eq:LR_alg}
\begin{split}
    \mathcal{L}_{n}&:=\alpha(\B_{n})\cap\cC_{n},\\
    \mathcal{R}_{n}&:=\alpha(\B_{n})\cap\cC_{n+1},
\end{split}
\eeq
then, obviously, we have $\mathcal{L}_{n}\otimes\mathcal{R}_{n-1}\subset\cC_{n}$. Moreover, it turns out that the reverse inclusion is also true.
\begin{lemma}[Theorem 4.1 of Ref.~\cite{Ranard_2022}]
    We have the following tensor-factorization of algebras
    \beq
    \begin{split}
        \cC_{n}&:=\A_{2n-1}\otimes\A_{2n}=\mathcal{L}_{n}\otimes\mathcal{R}_{n-1},\\
        \B_{n}&:=\A_{2n}\otimes\A_{2n+1}=\alpha^{-1}(\mathcal{L}_{n})\otimes\alpha^{-1}(\mathcal{R}_{n}).
    \end{split}
    \eeq
\end{lemma}

Based on this lemma, we can dfine the GNVW index of a QCA.
\begin{definition}[GNVW index of QCA]
    For a nearest-neighbor QCA $\alpha$ of a spin chain, its index is defined as\footnote{Our definition includes a minus sign compared to Ref.~\cite{Ranard_2022}.}
    \beq
    \begin{split}
       \ind(\alpha)&:=\frac{1}{2}(\log(\dim(\mathcal{R}_{n}))-\log(\dim(\A_{2n+1})))\\
       &=\frac{1}{2}(\log(\dim(\A_{2n}))-\log(\dim(\mathcal{L}_{n})))
    \end{split}
    \eeq
\end{definition}
The second equality is due to $\dim(\A_{2n-1}\otimes\A_{2n})=\dim(\mathcal{L}_{n}\otimes\mathcal{R}_{n-1})$. The GNVW index takes value in $\z[\{\log(p_{j})\}_{j\in J}]$ where $p_{j}$'s are prime divisors of $D$, the dimension of local Hilbert space, as explained in Eq.~\eqref{eq:GNVW_index}. It can be shown that $\ind:\G^{\QCA}\to\z[\{\log(p_{j})\}_{j\in J}]$ is a group homomorphism \cite{Gross_2012}.

As an example, the GNVW index of block-paritioned unitaries (BPU, see Eq.~\eqref{eq:BPU}) vanishes. To see this, one can re-group the lattice such that this BPU acts in an on-site way. In this case, $\alpha(\A_{k})=\A_{k}\,,\,\forall k\in\z$, so $\mathcal{L}_{n}=\A_{2n}$. As a result,
\beq
\ind(\alpha)=\frac{1}{2}(\log(\dim(\A_{2n}))-\log(\dim(\mathcal{L}_{n})))=0
\eeq
Recall that finite-depth quantum circuits are finite compositions of BPU's, since $\ind$ is a group homomorphism, they must have  a vanishing GNVW index as well.

In the case of shift $\tau$ (\ie translation by $+1$), $\mathcal{L}_{n}=\bbC$, so 
\beq
\ind(\tau)=\frac{1}{2}\log(\dim(\A_{2n}))=\log(D)
\eeq
where $D=\dim\cH_{2n}$.

We briefly comment on generalizing above definition to $\G^{lp}$. In Ref.~\cite{Ranard_2022}, it is proved that any LPA can be well approximated by a sequence of QCA, \ie for each $\alpha\in\G^{lp}$, there exists a sequence $\{\beta_{n}\}_{n=1,2,...}$ such that
\beq
\alpha=\lim_{n\to\infty}\beta_{n}
\eeq
so one can define
\beq
\ind(\alpha):=\lim_{n\to\infty}\ind(\beta_{n})
\eeq
It can be shown that $\ind(\alpha)$ is well-defined, \ie it does not depend on the choice of $\beta_{n}$ (see Ref.~\cite{Ranard_2022} for more information).

\subsection{GNVW index for twisted QCA and twisted LPA}\label{sec:twisted_index}

With the above background, in this sub-section, we generalize the notion of GNVW index to $\tG^{\QCA}$ and $\tG^{lp}$. We will focus on $\tG^{\QCA}$ first and then extend the discussion to $\tG^{lp}$. 

As noted in the main text (see Eq.~\eqref{eq:short_ex_seq}), for any anti-linear $\tilde{\alpha}\in\tG^{\QCA}$, we always have
\beq
\tilde{\alpha}=\alpha K,\quad\alpha\in\G^{\QCA}.
\eeq

We define the GNVW index of $\tilde\alpha$ as follows.
\begin{definition}[GNVW index for $\tG^{\QCA}$]\label{def:GNVW_twisted}
    Let $\tilde{\alpha}\in\tG^{\QCA}$, we define the GNVW index to be
    \beq
    \ind(\tilde{\alpha}):=
    \begin{cases}
        \ind(\tilde{\alpha}),\quad\tilde{\alpha}\,\text{is linear}\\
        \ind(\alpha),\quad\tilde{\alpha}\,\text{is anti-linear}\\
    \end{cases}
    \eeq
    In the second case, we have used $\tilde{\alpha}=\alpha K$ for some $\alpha\in\G^{\QCA}$.
\end{definition}

We have to show this definition is well-defined (\ie it does not depend on how we choose $\alpha$ and $K$ for a given $\tilde\alpha$) and it gives a group homomorphism from $\tG^{\QCA}\to\z[\{\log(p_{j})\}_{j\in J}]$ (we only know that it is homomorphism from $\G^{\QCA}\to\z[\{\log(p_{j})\}_{j\in J}]$ so far). 

To proceed, we need the following lemma.
\begin{lemma}\label{lemma:index_invariance}
    For $K$ defined in Eq. \eqref{eq:conjugation_operator}, we have
    \beq
    \ind(K\alpha K^{-1})=\ind(\alpha)
    \eeq
    for any $\alpha\in\G^{\QCA}$.
\end{lemma}
\begin{proof}
    Without loss of generality, we assume that $\alpha$ is a nearest-neighbor QCA by re-grouping the lattice. We also adapt the notations introduced in Eq.~\eqref{eq:block_alg} and \eqref{eq:LR_alg}. 
    
    Let us recall that $K$ is actually a linear map $K:\bar{\A}^{ql}\to\A^{ql}$, so by $K \alpha K^{-1}$ we really mean $K\bar{\alpha}K^{-1}$, where $\bar{\alpha}:\bar{\A}^{ql}\to\bar{\A}^{ql}$ is induced by $\alpha$. Since $K$ is on-site, we have $K^{-1}(\A_{n})= \bar{\A}_{n}$ for each $n\in\z$. Moreover, $K\bar{\alpha}K^{-1}$ is also nearest-neighbor QCA. By the definition of GNVW index, we consider
    \beq
    \begin{split}
        \mathcal{L}_{n}'&:=K\bar{\alpha}K^{-1}(\B_{n})\cap\cC_{n}=K\bar{\alpha}(\bar{\B}_{n})\cap\cC_{n}\\
        \mathcal{R}_{n}'&:=K\bar{\alpha}K^{-1}(\B_{n})\cap\cC_{n+1}=K\bar{\alpha}(\bar{\B}_{n})\cap\cC_{n+1}
    \end{split}
    \eeq
    We only have to show $\dim(\mathcal{L}_{n}')=\dim(\mathcal{L}_{n}),\,\forall\,n\in\z$. To see it, note that
    \beq
    K(\mathcal{L}_{n}')=\bar{\alpha}(\bar{B}_{n})\cap\bar{\cC}_{n}=\bar{\mathcal{L}}_{n}
    \eeq
    So we have $\dim(\mathcal{L}_{n}')=\dim(\bar{\mathcal{L}}_{n})=\dim(\mathcal{L}_{n})$. As a consequence, we have
    \beq
    \ind(K\alpha K^{-1})=\frac{1}{2}(\log(\dim(\A_{2n})-\log(\dim(\mathcal{L}'_{n}))=\ind(\alpha)
    \eeq
\end{proof}
Now we are ready to show that $\ind$ is well-defined.
\begin{theorem}\label{thm:well_defined}
    The GNVW index map defined in Definition \ref{def:GNVW_twisted} is well-defined, \ie it does not depend on the choice of $\alpha$ and $K$ for a given $\tilde\alpha$.
\end{theorem}
\begin{proof}
    Suppose we choose a different complex conjugation denoted by $K'$. Since $K'\in\tG^{\QCA}$ is anti-linear, there exists $\eta\in\G^{\QCA}$ such that
    \beq
    K'=\eta K
    \eeq
    On the other hand, $(K')^{2}=1$ by assumption. So we have
    \beq
    0=\ind((K')^{2})=\ind(\eta K\eta K)=\ind(\eta)+\ind(K\eta K^{-1})
    \eeq
    where we have used that $\ind$ is a homomorphism on $\G^{\QCA}$ and $K=K^{-1}$.
    By Lemma \ref{lemma:index_invariance}, we conclude that
    \beq
    \ind(\eta)=0
    \eeq
    By the structure theorem of QCA, this means $\eta$ is a circuit.
    So the choice of complex conjugation is unique up to a circuit.

    Now let $\tilde{\alpha}\in\tG^{\QCA}$ be anti-linear. Suppose $\alpha'K'$ is its another decomposition. That is,
    \beq
    \tilde{\alpha}=\alpha K=\alpha' K'=\alpha\eta K
    \eeq
    To show that the Definition \ref{def:GNVW_twisted} is well-defined, we have to show $\ind(\alpha')=\ind(\alpha\eta)=\ind(\alpha)$. This is indeed true because $\eta\in\G^{\QCA}$ and $\ind(\eta)=0$.

    Therefore, Definition \ref{def:GNVW_twisted} is well-defined.
\end{proof}
Now we are ready to show that $\ind$ is a homomorphism from $\tG^{\QCA}$ to $\z[\{\log(p_{j})\}_{j\in J}]$.
\begin{theorem}
    The GNVW index map in Definition \ref{def:GNVW_twisted} is a group homomorphism from $\tG^{\QCA}$ to $\z[\{\log(p_{j})\}_{j\in J}]$.
\end{theorem}
\begin{proof}
    There are three nontrivial cases. All automorphisms with $\tilde{\cdot}$ will be anti-linear below.

    First, let us consider $\alpha\tilde{\beta}$, where $\alpha\in\G^{\QCA}$ and $\beta\in\tG^{\QCA}$. We need to show 
    \beq
    \ind(\alpha\tilde{\beta})=\ind(\alpha)+\ind(\tilde{\beta})
    \eeq
    By writing $\tilde{\beta}=\beta K$ for some $\beta\in\G^{\QCA}$, one gets
    \beq
    \ind(\alpha\tilde{\beta})=\ind(\alpha\beta K)=\ind(\alpha\beta)
    \eeq
    and the desired result follows since $\ind(\alpha\beta)=\ind(\alpha)+\ind(\beta)$.

    Secondly, let us consider $\tilde{\alpha}\beta$. Writing $\tilde{\alpha}=\alpha K$, we have
    \beq
    \tilde{\alpha}\beta=\alpha K\beta=\alpha (K\beta K^{-1})K.
    \eeq
    So, by definition
    \beq
    \begin{split}
        \ind(\tilde{\alpha}\beta)&=\ind(\alpha (K\beta K^{-1}))\\
        &=\ind(\alpha)+\ind(K\beta K^{-1})\\
        &=\ind(\alpha)+\ind(\beta)\\
        &=\ind(\tilde{\alpha})+\ind(\beta).
    \end{split}
    \eeq
    In the third equality we have used Lemma \ref{lemma:index_invariance}.

    The last case $\tilde{\alpha}\tilde{\beta}$ can be treated in the same way as the second case.
\end{proof}

The above discussion gives a definition of the GNVW index of a twisted QCA, and shows that this definition is well-defined and it is a homomorphism from $\tG^{\QCA}$ to $\z[\{\log(p_{j})\}_{j\in J}]$.

The index theory can be constructed similarly for $\tG^{lp}$. The only difference is that in Theorem \ref{thm:well_defined}, the statement that “the choice of complex conjugation is unique up to a circuit” should be replaced by “the choice of complex conjugation is unique to a time evolution generated by some local Hamiltonian”.

\section{Anomaly index}\label{app:twisted_cocycle}

In this appendix, we present the proofs of various aspects of the anomaly index, which ensures that the anomaly index defined in the main text is a valid concept.

\subsection{Decomposition of $\tG^{lp}$}\label{sec:decomposition}
 
In this subsection, we establish Proposition \ref{prop:decomp_twisted}. To this end, we first review a few lemmas about decomposition in $\G^{lp}$ from Ref.~\cite{kapustin2024anomalous}. Proposition \ref{prop:decomp_twisted} is the generalization of these lemmas to $\tG^{lp}$.

\begin{lemma}[Lemma 2.1 of Ref.~\cite{kapustin2024anomalous}]\label{lemma:LPA_decom}
Let $\alpha\in\G^{lp}$ with vanishing GNVW index, then 
\beq
\alpha=\alpha_{<0}\alpha_{0}\alpha_{\geqslant 0}
\eeq
where $\alpha_{<0}\in\G^{lp}_{<0}$, $\alpha_{0}\in\G^{lp}_{0}$ and $\alpha_{\geqslant 0}\in\G^{lp}_{\geqslant 0}$. The converse is also true.
\end{lemma}

Let us explain the intuition behind the above result. Given the decomposition Eq. \eqref{eq:algebra_decomp}, one may want to restrict the symmetry action $\alpha$ to each half chain, which gives the $\alpha_{\geqslant 0}$ and $\alpha_{<0}$ parts. However, generically $\alpha$ can expand the support of an operator. If an operator is supported on one of the two half chains, under $\alpha$ it will generically acquire some support on the other half chain. But $\alpha_{\geqslant 0}$ and $\alpha_{<0}$ cannot achieve this, so the $\alpha_{0}$ part is also expected. 

\begin{lemma}[Lemma 2.2 of Ref. \cite{kapustin2024anomalous}]\label{lemma:non_unique_decomp}
    Suppose $\alpha\in\G^{lp}$ has a vanishing GNVW index and it admits two different decomposition
    \beq
    \alpha=\alpha_{<0}\alpha_{0}\alpha_{\geqslant 0}=\alpha'_{<0}\alpha'_{0}\alpha'_{\geqslant 0}
    \eeq
    then $\alpha_{<0}(\alpha'_{<0})^{-1}\in\G^{lp}_{0}$ and $\alpha_{\geqslant 0}(\alpha'_{\geqslant 0})^{-1}\in \G^{lp}_{0}$.
\end{lemma}

Now we turn to $\tG^{lp}$. Note that one can always fix $K$ (which is chosen to be on-site, as in Eq.~\eqref{eq:conjugation_operator}), so that for any anti-linear $\tilde{\alpha}\in\tG^{lp}$, we have
\beq
\tilde{\alpha}=\alpha K,\quad\alpha\in\G^{lp}
\eeq
Assuming $\ind(\tilde{\alpha})=\ind(\alpha)=0$, the following propositions are true.
\begin{proposition}
    Let $\tilde{\alpha}\in\tG^{lp}$ with vanishing GNVW index. Then it admits a decomposition
    \beq\label{eq:decompcondition}
    \tilde{\alpha}=\tilde{\alpha}_{<0}\alpha_{0}\tilde{\alpha}_{\geqslant0}
    \eeq
    where $\tilde{\alpha}_{<0}\in\tG^{lp}_{<0}$, $\alpha_{0}\in \G_{0}^{lp}$ and $\tilde{\alpha}_{\geqslant0}\in\tG^{lp}_{\geqslant0}$.
    Moreover, $\tilde{\alpha}_{\geqslant 0}$ (resp. $\tilde{\alpha}_{<0}$) is anti-linear on $\A^{ql}_{\geqslant0}$ (resp. $\A^{ql}_{<0}$) if and only if $\tilde{\alpha}$ is anti-linear. The converse is also true.
\end{proposition}
\begin{proof}
    The case where $\tilde{\alpha}$ is linear is addressed by Lemma \ref{lemma:LPA_decom}. In the following, we assume $\tilde{\alpha}$ is anti-linear.
    By Lemma \ref{lemma:LPA_decom}, we have
    \beq
    \alpha=\alpha_{<0}\alpha'_{0}\alpha_{\geqslant0}
    \eeq
    Since $K$ is on-site, one can define
    \beq
    \begin{split}
        K_{\geqslant0}&:=\prod_{n\geqslant0}K_{n}\\
        K_{<0}&:=\prod_{n<0}K_{n}
    \end{split}
    \eeq
    Then one defines
    \beq
    \begin{split}
        \tilde{\alpha}_{\geqslant0}&:=\alpha_{\geqslant0}K_{\geqslant0}\\
        \alpha_{0}&:=K_{<0}\triangleright\alpha_{0}'=K_{<0}^{-1}\alpha_{0}'K_{<0}\\
        \tilde{\alpha}_{<0}&:=\alpha_{<0}K_{<0}
    \end{split}
    \eeq
    Using that $K_{<0}\alpha_{\geqslant 0}=K_{<0}\alpha_{\geqslant 0}K_{<0}^{-1}K_{<0}=\alpha_{\geqslant 0}K_{<0}$, one can see that $\tilde{\alpha}=\tilde{\alpha}_{<0}\alpha_{0}\tilde{\alpha}_{\geqslant0}$. This establishes the decomposition and the (anti-)linearity of $\tilde\alpha_{\geqslant 0}$ and $\tilde\alpha_{<0}$.

    Conversely, if $\tilde{\alpha}$ is anti-linear of the form of Eq.~\eqref{eq:decompcondition}, then we have 
    \beq
    \tilde{\alpha}\circ K=(\talpha_{<0}\circ K_{<0})\circ (K_{<0}^{-1}\alpha_{0} K_{<0})\circ(\talpha_{\geqslant0}\circ K_{\geqslant0})
    \eeq
    Note that $\talpha\circ K$, $K_{<0}^{-1}\alpha_{0} K_{<0}$, $\talpha_{<0}\circ K_{<0}$ and $\talpha_{\geqslant0}\circ K_{\geqslant0}$ are all linear LPA's. Thus $\ind(\talpha)=\ind(\talpha\circ K)=0$ is ensured by Lemma \ref{lemma:LPA_decom}.

\end{proof}

The second proposition characterizes the relation between different decompositions of $\tG^{lp}$. We state a stronger version and prove it here, which will be used later.
\begin{proposition}\label{prop:non-unique}
        Suppose $\tilde{\alpha}\in\tG^{lp}$ has a vanishing GNVW index and it can be decomposed as
    \beq
    \tilde{\alpha}_{<0}\alpha_{0}\tilde{\alpha}_{\geqslant0}=\tilde{\alpha}=\tilde{\beta}_{<n}\beta_{n}\tilde{\beta}_{\geqslant n}
    \eeq
    where $n\in\z_{\geqslant 0}$ is some lattice site and $\beta_{n}\in\G_{0}^{lp}$.
    Then there exists a quasi-local unitary $V$ supported on the non-negative half chain, such that
    \beq
    \tilde{\beta}_{\geqslant n}=\Ad_{V}\circ\tilde{\alpha}_{\geqslant0}
    \eeq
    if $\tilde\alpha$ is linear. If $\tilde\alpha$ is anti-linear, then there exists a quasi-local unitary $V$ supported on the non-negative half chain, such that
    \beq
    \tilde\beta_{\geqslant n}=\Ad_V\circ K_{[0, n)}\tilde\alpha_{\geqslant 0},
    \eeq
    where $K_{[0, n)}=\prod_{j=0}^{n-1}K_j$.

    A similar statement holds if $n\in\z_{< 0}$, with $K_{[0, n)}$ above replaced by $K_{[n, 0)}=\prod_{j=n}^{-1}K_j$.
\end{proposition}
\begin{proof}
    Without loss of generality, we assume $n\geqslant 0$. Then
    $$\beta_n=\tilde{\beta}_{<n}^{-1}\tilde{\alpha}_{<0}\alpha_0\tilde\alpha_{\geqslant 0}\tilde\beta_{\geqslant n}^{-1}=\left(\tilde{\beta}_{<n}^{-1}\tilde{\alpha}_{<0}\tilde{\alpha}_{\geqslant0}\tilde{\beta}_{\geqslant n}^{-1}\right)\left((\tilde{\alpha}_{\geqslant0}\tilde{\beta}_{\geqslant n}^{-1})^{-1}\alpha_0\tilde\alpha_{\geqslant 0}\tilde\beta_{\geqslant n}^{-1}\right)\in\G_{0}^{lp}.$$
    Since $(\tilde{\alpha}_{\geqslant0}\tilde{\beta}_{\geqslant n}^{-1})^{-1}\alpha_0\tilde\alpha_{\geqslant 0}\tilde\beta_{\geqslant n}^{-1}\in\G_{0}^{lp}$, we get that $\tilde{\beta}^{-1}_{<n}\tilde{\alpha}_{<0}\tilde{\alpha}_{\geqslant0}\tilde{\beta}_{\geqslant n}^{-1}\in\G^{lp}_{0}$. Thus, according to Lemma 3.1 of Ref.~\cite{Lance68auto}, for any $\epsilon>0$, there exists $R_{\epsilon}>0$ such that for any $A\in\A^{ql}_{[R_{\epsilon},\infty)}$, we have
    \beq
    ||\tilde\beta^{-1}_{<n}\tilde\alpha_{<0}\tilde{\alpha}_{\geqslant0}\tilde{\beta}_{\geqslant n}^{-1}(A)-A||=||\tilde{\alpha}_{\geqslant0}\tilde{\beta}_{\geqslant n}^{-1}(A)-A||<\epsilon||A||
    \eeq
    where we have assumed $R_{\epsilon}>n$, which is always possible. 
    
    If $\tilde\alpha$ is linear on $\A^{ql}$, then $\tilde\beta^{-1}_{<n}$, $\tilde\alpha_{<0}$, $\tilde\alpha_{\geqslant 0}$ and $\tilde\beta^{-1}_{\geqslant n}$ are all linear on the algrebras of their respective supports. Using Lemma 3.1 of Ref.~\cite{Lance68auto} again, we conclude that there is a quasi-local unitary operator $V$ (supported on non-negative half chain) such that
    \beq
    \tilde{\beta}_{\geqslant n}=\Ad_{V}\circ\tilde{\alpha}_{\geqslant0}.
    \eeq

    On the other hand, if $\tilde\alpha$ is anti-linear on $\A^{ql}$, then $\tilde\alpha_{\geqslant 0}$ is anti-linear on $\A^{ql}_{[0, \infty)}$ and $\tilde\beta^{-1}_{\geqslant n}$ is anti-linear on $\A^{ql}_{[n, \infty)}$. So $K_{[0, n)}\tilde\alpha_{\geqslant 0}\tilde\beta^{-1}_{\geqslant n}$ is linear on $\A^{ql}_{\geqslant 0}$. Moreover, because $R_\epsilon>n$ by assumption,
    \beq
    ||K_{[0, n)}\tilde{\alpha}_{\geqslant0}\tilde{\beta}_{\geqslant n}^{-1}(A)-A||=||\tilde{\alpha}_{\geqslant0}\tilde{\beta}_{\geqslant n}^{-1}(A)-A||<\epsilon||A||
    \eeq
    Now we can use Lemma 3.1 of Ref. ~\cite{Lance68auto} again, and conclude that there is a quasi-local unitary operator $V$ (supported on the non-negative half chain) such that
    \beq
    \tilde\beta_{\geqslant n}=\Ad_V\circ K_{[0, n)}\tilde\alpha_{\geqslant 0}.
    \eeq

\end{proof}

A special case of the above proposition is when $n=0$. In this case, $K_{[0, 0)}$ should be viewed as identity, and this proposition reduces to Proposition \ref{prop:nonunique_decomp} in the main text.

\subsection{Twisted 3-cocycle condition} \label{subapp: 3-cocycle equation}

We now check that $\omega$ defined in Eq.~\eqref{eq:anomaly_index} is indeed a twisted 3-cocycle, and thus prove the first part of Theorem \ref{theorem:anomaly_index}.

To simplify the notations, we will denote $\tilde{\beta}_{g}:=\tilde{\alpha}_{\geqslant0}(g)$ in this section. Then, by using the definition Eq.~\eqref{eq:anomaly_index} repeatedly, we find that the cocycle condition follows from
\beq
\begin{split}
    &\omega(gh,k,l)\omega(g,h,kl)(\Ad_{V(g,h)}\circ\tilde{\beta}_{gh})(V(k,l))\tilde{\beta}_{g}(V(h,kl))V(g,hkl)\\
    =&V(g,h)\omega(gh,k,l)\tilde{\beta}_{gh}(V(k,l)) V(g,h)^{-1}\omega(g,h,kl)\tilde{\beta}_{g}(V(h,kl))V(g,hkl)\\
    =&V(g,h)\omega(gh,k,l)\tilde{\beta}_{gh}(V(k,l))V(gh,kl)\\
    =&V(g,h)V(gh,k)V(ghk,l)\\
    =&\omega(g,h,k)\tilde{\beta}_{g}(V(h,k))V(g,hk)V(ghk,l)\\
    =&\omega(g,h,k)\tilde{\beta}_{g}(V(h,k))\omega(g,hk,l)\tilde{\beta}_{g}(V(hk,l))V(g,hkl)\\
    =& \omega(g,h,k)\omega(g,hk,l)\tilde{\beta}_{g}(\omega(h,k,l)\tilde{\beta}_{h}(V(k,l))V(h,kl))V(g,hkl)\\
    =&\omega(g,h,k)\omega(g,hk,l)(g\triangleright \omega(h,k,l))(\Ad_{V(g,h)}\circ\tilde{\beta}_{gh})(V(k,l))\tilde{\beta}_{g}(V(h,kl))V(g,hkl)
\end{split}
\eeq
where $g\triangleright\omega(h, k, l)=\omega(h, k, l)$ if $g$ is linear, and $g\triangleright \omega(h,k,l)=\overline{\omega(h,k,l)}$ if $\tilde{\alpha}(g)$ is anti-linear.

\subsection{Independence of all artifical choices}
We first show that anomaly index is independent of decomposition we use in Eq.~\eqref{eq:tLPA_decomposition}.

\begin{proof}[Proof of the independence of decomposition]
    To simplify the notations, we denote $\tilde{\alpha}_{\geqslant 0}$ by $\tilde{\beta}_{g}$.

    By Proposition \ref{prop:nonunique_decomp}, different choices of $\beta_{g}$ are related by adjoint action of $U_{g}\in\cU^{ql}$, \ie
    \beq
    \tilde{\beta}_{g}'=\Ad_{U_{g}}\tilde{\beta}_{g}
    \eeq
    From
    \beq
    \begin{split}
        \tilde{\beta}_{g}\tilde{\beta}_{h}&=\Ad_{V(g,h)}\tilde{\beta}_{gh}\\
        \tilde{\beta}'_{g}\tilde{\beta}'_{h}&=\Ad_{V'(g,h)}\tilde{\beta}'_{gh}
    \end{split}
    \eeq
    One can solve
    \beq
    V'(g,h)=U_{g}\tilde{\beta}_{g}(U_{h})V(g,h)U_{gh}^{-1}
    \eeq
    up to a phase factor $\eta(g,h)\in\U$. The effect of this phase factor is to shift $\omega$ by a 3-coboundary $\tilde{\delta}\eta$, so we will not be concerned about this ambiguity.

    We have
    \beq
    \begin{split}
        V'(gh,k)&=U_{gh}\tilde{\beta}_{gh}(U_{k})V(gh,k)U_{ghk}^{-1}\\
        V'(g,hk)^{-1}&=U_{ghk}V(g,hk)^{-1}\tilde{\beta}_{g}(U_{hk})^{-1}U_{g}^{-1}\\
        \tilde{\beta}'_{g}(V'(h,k))^{-1}&=U_{g}\tilde{\beta}_{g}(U_{hk})\tilde{\beta}_{g}(V(h,k))^{-1}\tilde{\beta}_{g}\tilde{\beta}_{h}(U_{k})^{-1}\tilde{\beta}_{g}(U_{h})^{-1}U_{g}^{-1}
    \end{split}
    \eeq
    Now using $\tilde{\beta}_{gh}=\Ad_{V(g,h)}\tilde{\beta}_{g}\tilde{\beta}_{h}$ in $V'(gh,k)$, we obtain
    \beq
    \omega'(g,h,k)=U_{g}\tilde{\beta}_{g}(U_{h})\tilde{\beta}_{g}\tilde{\beta}_{h}(U_{k})\omega(g,h,k)\tilde{\beta}_{g}\tilde{\beta}_{h}(U_{k})^{-1}\tilde{\beta}_{g}(U_{h})^{-1}U_{g}^{-1}=\omega(g,h,k)
    \eeq

    So the anomaly index is independent of the decomposition we use in Eq.~\eqref{eq:tLPA_decomposition}.
    
\end{proof}

Combining Proposition \ref{prop:non-unique} and an argument similar to the above proof, we can see that $\omega$ is also independent of the site at which we perform the decomposition of Eq.~\eqref{eq:tLPA_decomposition}.

\section{Proof of Lemma \ref{lemma: root lemma}} \label{app: proving the root lemma}

In this section, we utilize some powerful tools to show Lemma \ref{lemma: root lemma}, which is a generalization of Lemma E.5 of Ref.~\cite{Liu2024LRLSM}. Lemma \ref{lemma: root lemma} is restated below.

\begin{lemma}\label{lemma:typeI_KS}
    Let $\talpha:G\to \tG^{lp}$ be a symmetry action and $\psi$ be a type-I factor state that satisfies the split property. If $\psi=\psi_{\talpha(g)},\forall g\in G$ (\ie $\psi$ is invariant under $G$), then the anomaly index $\tilde{\omega}\in\rH^{3}_{\varphi}(G,\U)$ vanishes.
\end{lemma}
\begin{proof}
    Recall that for a general factor state, we say that $\psi$ splits if $\psi\sim\psi_{<0}\otimes\psi_{\geqslant0}$, where $\sim$ means quasi-equivalence. We then have
    \beq
    \psi_{<0}\otimes \psi_{\geqslant0}\sim\psi=\psi_{\talpha(g)}\sim (\psi_{<0\talpha_{<0}(g)} )\otimes(\psi_{\geqslant0\talpha_{\geqslant0}(g)})
    \eeq
    One then deduces that $\psi_{\geqslant 0}\sim\psi_{\geqslant 0\talpha_{\geqslant 0}(g)}$. So $\psi\sim\psi_{<0}\otimes\psi_{\geqslant 0}\sim\psi_{<0}\otimes\psi_{\geqslant 0\talpha_{\geqslant0}(g)}\sim\psi_{\talpha_{\geqslant 0}(g)}$. 

    For now, we assume $\talpha(g)$ is linear, and the anti-linear case will be treated shortly. In this case, $\psi_{\talpha_{\geqslant 0}(g)}=\psi\circ\talpha_{\geqslant 0}(g)$.
    Let $(\cH_{\psi},\pi_{\psi},|\psi\ra)$ be the GNS triple of $\psi$. Then the GNS representation of $\psi\circ\talpha_{\geqslant0}(g)$ is given by $\pi_{\psi}\circ\talpha_{\geqslant 0}(g)$ acting on $\cH_{\psi}$.
    Because $\psi\sim\psi\circ\talpha_{\geqslant 0}(g)$, by Definition \ref{def: quasi-equivalence}, there exists a $*$-isomorphism $F_{g}:\cM_{\psi}\to\cM_{\psi\circ\talpha_{\geqslant0}(g)}$ such that
    \beq
    F_{g}(\pi_{\psi}(a))=\pi_{\psi}([\talpha_{\geqslant 0}(g)](a)), \forall a\in\A^{ql}.
    \eeq
    Moreover, $\pi_{\psi}(\A^{ql})'=\pi_{\psi}(\talpha(g)_{\geqslant0}\A^{ql})'$ since $\talpha(g)_{\geqslant0}$ is a $*$-automorphism. This in turn means $\cM_{\psi}=\cM_{\psi\circ\talpha_{\geqslant0}(g)}$ as an abstract von Neumann algebra. Thus, $F_{g}$ can be thought of as a $*$-automorphism on $\cM_{\psi}$. By assumption, $\cM_{\psi}$ is a type-I factor acting on $\cH_{\psi}$, which means $\cM_{\psi}=\B(\cH)\otimes 1_{\cK}$ with $\cH$ and $\cK$ two other Hilbert spaces such that $\cH_{\psi}=\cH\otimes \cK$. Now using Lemma \ref{lemma:typeI}, which means there exists a map $\tilde{U}_{g}\in\cM_\psi$ such that
    \beq
    \pi_{\psi}(\talpha_{\geqslant0}(a)) = \tilde U(g)\pi_\psi(a)\tilde U(g)^{-1}
    \eeq

    If $\talpha(g)$ is anti-linear instead, $\psi\sim\psi_{\talpha_{\geqslant0}(g)}$ in fact means $\bar{\psi}\simeq\psi\circ\talpha_{\geqslant0}(g)$ as states of $\bar{\A}^{ql}$. The map $F_g$ still can be constructed as $*$-automorphism of $\bar{\cM}_{\psi}$, which is again given by a unitary operator $\tilde{U}(g)\in\bar{\cM}_{\psi}$. However, below we regard $\tilde{U}(g)$ as an anti-linear map $\tilde{U}(g):\cH_{\psi}\to\bar{\cH}_{\psi}$.

    In summary, we have shown that
    there exists an operator $\tilde{U}(g)$, such that
    \beq\label{eq:Jordan_GNS}
    \pi_{\psi_{\tilde{\alpha}_{\geqslant0}(g)}}(A)=\tilde{U}(g)\pi_\psi(A)\tilde{U}(g)^{-1}, \forall A\in\A^{ql}.
    \eeq
    where $\tilde{U}$ is unitary (resp. anti-unitary) if $\tilde{\alpha}$ is linear (resp. anti-linear).
    
    According to Eq.~\eqref{eq:near_homomoprhism2}, we have
    \beq
    \pi_\psi(\Ad_{V(g,h)}A)=\pi_\psi(\tilde{\alpha}_{\geqslant0}(g)\tilde{\alpha}_{\geqslant0}(h)\tilde{\alpha}_{\geqslant0}(gh)^{-1}A)=\pi_{\psi_{\talpha_{\geqslant 0}(g)\talpha_{\geqslant 0}(h)}\talpha_{\geqslant 0}(gh)^{-1}}(A), \forall A\in\A^{ql}.
    \eeq
    By Eq.~\eqref{eq:Jordan_GNS},
    \beq\label{eq:adjoint_GNS}
    \pi_\psi(V(g,h))\pi_\psi(A)\pi_\psi(V(g,h))^{-1}=\tilde{U}(g)\tilde{U}(h)\tilde{U}(gh)^{-1}\pi_\psi(A)\tilde{U}(gh)\tilde{U}(h)^{-1}\tilde{U}(g)^{-1}, \forall A\in\A^{ql}.
    \eeq
    We emphasize that $\tilde{U}(g)$ does not need to be a homomorphism. 
    
    The above Eq.~\eqref{eq:adjoint_GNS} means that 
    \beq \label{eqapp: emergence of a phase}
    \eta(g,h):=\pi_\psi(V(g,h))^{-1}\tilde{U}(g)\tilde{U}(h)\tilde{U}(gh)^{-1}
    \eeq
    commutes with every element in $\pi_\psi(A)$. Thus $\eta(g, h)\in\pi_\psi(\A^{ql})'\simeq\cM_{\psi}'$. On the other hand, $\tilde U(g)\tilde U(h)\tilde U(gh)^{-1}\in\cM_\psi$ (no matter whether $g$ and $h$ are unitary or anti-unitary symmetries, because $\talpha_{\geqslant 0}(g)\talpha_{\geqslant 0}(h)\talpha_{\geqslant 0}(gh)^{-1}$ is always linear), and $\pi_\psi(V(g, h))^{-1}\in\pi_\psi(\A^{ql})\subseteq \pi_\psi(\A^{ql})''=\cM_\psi$. Therefore, $\eta(g, h)\in\cM_\psi$. So $\eta(g,h)\in\cM_{\psi}\cap\cM_{\psi}'\simeq\mathbb{C}\cdot \id$, where we have used Definition \ref{definition: factor}. Clearly, $\eta(g, h)$ is unitary, so $\eta(g, h)\in\U$.
    
    Similar to Eq. \eqref{eqapp: emergence of a phase}, we have
    \begin{widetext}
        \beq
        \begin{split}
    \pi(V(g,h))&=\eta(g,h)^{-1}\tilde{U}(g)\tilde{U}(h)\tilde{U}(gh)^{-1}\\
    \pi(V(gh,k))&=\eta(gh,k)^{-1}\tilde{U}(gh)\tilde{U}(k)\tilde{U}(ghk)^{-1}\\
    \pi(V(g,hk))^{-1}&=\eta(g,hk)\tilde{U}(ghk)\tilde{U}(hk)^{-1}\tilde{U}(g)^{-1}\\
    \pi(\tilde{\alpha}_{\geqslant0}(g)\triangleright V(h,k))^{-1}&=\tilde{U}(g)\eta(h,k)\tilde{U}(hk)\tilde{U}(k)\tilde{U}(h)^{-1}\tilde{U}(g)^{-1}
        \end{split}
        \eeq
    \end{widetext}
    where all $\eta$'s above are $U(1)$ phase factors.
    
    Combining the above equation with Eq. \eqref{eq:anomaly_index}, it is straightforward to verify that
    \beq
    \omega(g,h,k)=(\tilde{\delta}^{(3)}\eta)(g,h,k)
    \eeq
    where the coboundary map $\tilde{\delta}^{(3)}$ is defined in Eq. \eqref{eq: defining coboundary}.
    Therefore, the anomaly index $\omega$ is trivial in $\rH^{3}_{\varphi}(G;\U)$.
    
\end{proof}

\section{Eq. \eqref{eq: anomaly index for anti-unitary translation} represents a nontrivial cohomology element} \label{app: nontrivial anomaly}

In this appendix, we show that Eq. \eqref{eq: anomaly index for anti-unitary translation} represents a nontrivial element in $H^3(\z_2^x\times\z_2^z\times\z^T; \U)$, by showing that a topological invariant evaluated on this cocycle gives $-1$.

For any cocycle $\Omega(g_1, g_2, g_3)$ with $g_{1, 2, 3}\in \z_2^x\times\z_2^z\times\z^T$, consider the following quantity{\footnote{We thank Shang-Qiang Ning for constructing this invariant.}}:
\beq \label{eq: topological invariant}
I[\Omega]=\frac{\Omega(T_2, X, Z)\Omega(X, Z, T_2)}{\Omega(X, T_2, Z)}\frac{\Omega(Z, T_2, X)}{\Omega(T_2, Z, X)\Omega(Z, X, T_2)}\Omega(X, XZ, X)\Omega(X, X, Z)\Omega(Z, X, X)\Omega(Z, e, Z),
\eeq
where $e$ is the identity element of $\z_2^x\times\z_2^z\times\z^T$. By straightforward calculations, one can verify that this quantity is a topological invariant of $H^3(\z_2^x\times\z_2^z\times\z^T; \U)$, \ie $I[\Omega]$ does not change under any coboundary transformations.

Clearly, for cocycles that represent the trivial element in $H^3(\z_2^x\times\z_2^z\times\z^T; \U)$, this invariant evaluates to 1, because one can take $\Omega(g_1, g_2, g_3)=1$ in this case. Now substituting Eq. \eqref{eq: anomaly index for anti-unitary translation} into Eq. \eqref{eq: topological invariant}, we find that this invariant evaluates to $-1$. Therefore, the cocycle in Eq. \eqref{eq: anomaly index for anti-unitary translation} represents a nontrivial element in $H^3(\z_2^x\times\z_2^z\times\z^T; \U)$.

\bibliography{lib.bib}

\end{document}